\documentclass[11pt,oneside]{article}

\usepackage{authblk}

\usepackage{amsthm}
\usepackage{amsmath,latexsym,amssymb,verbatim,graphicx,enumitem}
\usepackage{quantum}

\usepackage[normalem]{ulem}
\usepackage{color}

\usepackage[backend=bibtex]{biblatex}
\bibliography{historyH}
\DefineBibliographyStrings{english}{in = {}}

\usepackage[colorlinks,linktocpage]{hyperref}
\usepackage[capitalise]{cleveref}

\makeatletter
\renewcommand\section{%
  \@startsection{section}{1}{0pt}%
  {-\baselineskip}{.5\baselineskip}%
  {\normalfont\Large\bfseries\raggedright}}
\renewcommand\subsection{%
  \@startsection{subsection}{2}{0pt}%
  {-\baselineskip}{.1\baselineskip}%
  {\large\bfseries\raggedright}}
\renewcommand\subsubsection{%
  \@startsection{subsection}{2}{0pt}%
  {-\baselineskip}{.1\baselineskip}%
  {\normalfont\scshape\raggedright}}
\makeatother

\newtheorem{thm}{Theorem}
\newtheorem*{thm*}{Theorem}
\newtheorem{assum}[thm]{Assumption}
\newtheorem*{assum*}{Assumption}
\newtheorem{definition}[thm]{Definition} 
\newtheorem{prop}[thm]{Proposition}
\newtheorem*{prop*}{Proposition}
\newtheorem{lem}[thm]{Lemma}
\newtheorem*{lem*}{Lemma}

\newtheorem*{fact*}{Fact}

\newtheorem*{cor*}{Corollary}

\newtheorem{rem}[thm]{Remark}
\newtheorem*{rem*}{Remark}

\crefname{thm}{Theorem}{Theorems}
\Crefname{thm}{Theorem}{Theorems}
\crefname{assum}{Assumption}{Assumptions}
\Crefname{assum}{Assumption}{Assumptions}
\crefname{prop}{Proposition}{Propositions}
\Crefname{prop}{Proposition}{Propositions}
\crefname{lem}{Lemma}{Lemmas}
\Crefname{lem}{Lemma}{Lemmas}
\crefname{cor}{Corollary}{Corollaries}
\Crefname{cor}{Corollary}{Corollaries}
\crefname{rem}{Remark}{Remarks}
\Crefname{rem}{Remark}{Remarks}
\crefname{equation}{eq.}{eqs.}
\Crefname{equation}{eq.}{eqs.}

\def\be{\begin{equation}}
\def\ee{\end{equation}}
\def\ben{\begin{eqnarray}}
\def\een{\end{eqnarray}}

\newcommand\cP{\mathcal{P}}
\newcommand\cR{\mathcal{R}}
\newcommand\cM{\mathcal{M}}
\newcommand\cS{\mathcal{S}}
\newcommand\cC{\mathcal{C}}

\newcommand{\Vpoly}{q}
\newcommand{\Ppoly}{q_1}

\newcommand\cA{\mathcal{A}}
\newcommand\cF{\mathcal{F}}
\newcommand\NN{\mathbb{N}}
\newcommand\constantC{c}

\title{History-state Hamiltonians are critical}
\author[1,2]{Carlos E. Gonz\'alez-Guill\'en\thanks{carlos.gguillen@upm.es}}
\affil[1]{Departamento de Matem\'atica Aplicada a la Ingenier\'ia Industrial, Universidad Polit\'ecnica de Madrid, Spain}
\affil[2]{IMI, Universidad Complutense de Madrid, Spain}
\author[3]{Toby S.\ Cubitt\thanks{t.cubitt@ucl.ac.uk}}
\affil[3]{Department of Computer Science, University College London, UK}

\date{}

\begin{document}

\maketitle

\begin{abstract}
  All Hamiltonian complexity results to date have been proven by constructing a local Hamiltonian whose ground state -- or at least some low-energy state -- is a ``computational history state'', encoding a quantum computation as a superposition over the history of the computation.
  We prove that all history-state Hamiltonians must be critical.
  More precisely, for any circuit-to-Hamiltonian mapping that maps quantum circuits to local Hamiltonians with low-energy history states, there is an increasing sequence of circuits that maps to a growing sequence of Hamiltonians with spectral gap closing at least as fast as $O(1/n)$ with the number of qudits $n$ in the circuit.
  This result holds for very general notions of history state, and also extends to quasi-local Hamiltonians with exponentially-decaying interactions.

  This suggests that QMA-hardness for gapped Hamiltonians (and also BQP-completeness of adiabatic quantum computation with constant gap) either require techniques beyond history state constructions,
  or gapped Hamiltonians cannot be QMA-hard (respectively, BQP-complete).
\end{abstract}



\section{Introduction}
The field of Hamiltonian complexity has produced a multitude of results concerning the computational complexity of various questions about quantum many-body Hamiltonians.
The best-known is the problem of estimating the ground state energy of local Hamiltonians (the ``Local Hamiltonian problem'').
This problem was proven complete for the class QMA for ever-simpler families of Hamiltonians: first for 5-local Hamiltonians by Kitaev\cite{Kitaev_book}, then 3-local\cite{KR}, 2-local\cite{KKR}, 2D qubit lattices\cite{Oliveira-Terhal}, 1D spin chains\cite{AGIK}, 1D translationally-invariant,\footnote{Technically this result proves QMA\textsubscript{EXP}-completeness; this is more an artefact of the way the translationally-invariant Hamiltonian is specified than a fundamental difference in the result.}
spin chains\cite{Gottesman-Irani}, low-dimensional trans\-la\-tion\-al\-ly-invariant spin chains\cite{wheelbarrow}, restricted classes of interactions\cite{Biamonte-Love}, and others\cite{spacetime}.

What all these results have in common, is that they encode quantum computation in the ground state (or at least some low-energy state) by constructing a family of local Hamiltonians whose ground states are superpositions over the history of the computation, with a form similar to:
\vspace{-.5em}
\begin{equation}\label{eq:vanilla_history_state}
  \ket{\Psi_0} = \sum_{t=0}^T \alpha_t \ket{t}\ket{\psi_t},
\vspace{-.5em}
\end{equation}
where $\ket{t}$ indicates computation time (the ``clock register'') and $\ket{\psi_t}$ is the state of the computation at time $t$ (the ``computational register'').
Hamiltonians with ground states of this form are sometimes called ``history-state Hamiltonians''.

This idea of encoding computation in superposition originally dates back to Feynman\cite{Feynman}, later picked up and significantly developed by Kitaev\cite{Kitaev_book}.
Different history states differ substantially in the precise way computation is encoded in the computational register and, particularly, in the way time is encoded.
The local Hamiltonian constructions with these history states as ground states differ even more substantially.
But one feature common to all these various Hamiltonian constructions is that their spectral gaps close as $1/\poly(n)$, where $n$ is the number of qubits in the circuit.
In condensed-matter terminology, these Hamiltonians are ``critical'' in the thermodynamic limit, indicative of being at a phase transition.\footnote{The spectral gap necessarily vanishes at a quantum phase transition. Note, however, that the converse is not necessarily true: Hamiltonians with vanishing gap can also occur away from a phase transition.}
It is tempting to conjecture that the computational complexity occurs \emph{because} they are critical.

Here, we prove that this intuition is correct for all history-state Hamiltonians (see \cref{main}, below, for the precise technical statement):
\begin{thm}[Main result -- informal]
  Any $k$-local Hamiltonian that has a low-energy history state has a spectral gap that closes as $O(1/n)$ in the system size $n$.
\end{thm}

\section{Main results}
To state our results precisely, we must first define precisely what we mean by a ``history state'', and what class of Hamiltonians we are considering.
We would like these to be as general as possible, to give as strong a result as possible.

\subsection{Hamiltonian normalisation}\label{sec:normalisation}
We consider Hamiltonians $H = \sum_Z h_Z$ constructed out of local (or quasi-local) interactions $h_Z$, where $Z$ is the subset of qudits the term acts on non-trivially.
It is important that we normalise our Hamiltonians appropriately, otherwise there are simple strategies to amplify its spectral gap arbitrarily.
A trivial example is multiplying the entire Hamiltonian by a scalar that grows polynomially.
We want to choose the mildest normalisation condition that is meaningful.
In particular, we do not want to explicitly restrict the allowed types of interaction graph.

Simple counterexamples show that the normalisation common in condensed matter physics of requiring the local interactions to have constant strength, $\inftynorm{h_Z} = O(1)$, does not suffice here.
Duplicating each local term $h_Z$ in the Hamiltonian $m$ times will trivially increase the spectral gap by a factor of $m$.
Additionally requiring all local terms that act on the same set of qudits to be grouped into a single local term does not rescue things.
Increasing $k$ by 1 and adding $m$ ancillas allows us to achieve the same spectral gap amplification whilst by-passing the normalisation constraint.
Simply tensor each copy of the local term with a projector acting on a different ancilla: $H = \sum_Z\sum_{i=1}^m h_Z\ox\proj{0}_i$.
This issue does not arise in typical condensed matter systems, as they invariably involve interactions restricted to some lattice, in which case requiring constant-strength local interactions suffices.
But restricting to a lattice would exclude many known Hamiltonian complexity constructions, including Kitaev's original construction.

Another instructive family of examples are history-state Hamiltonians (unitarily equivalent to) $H = \Delta\ox\1$, where $\Delta$ is the graph Laplacian of the random walk on the $\log_2(T)$-dimensional unit hypercube.
The random walk on the unit hypercube has constant mixing time, so this Hamiltonian has constant spectral gap.
However, the $\log_2(T)$-dimensional hypercube graph has vertex degree $\log_2(T)$, implying that the number of local interactions acting on any given clock qubit diverges as $\log_2(T)$.

The appropriate normalisation for $k$-local Hamiltonians is to require the total strength of the interactions acting on any given qudit to be independent of the total number of qudits, and that is the condition under which our results apply (see \cref{sec:k-local}):
\begin{assum}\label{normalisation}
  There exists a constant $\gamma$ such that, for any qudit $q$,
  \begin{equation}
    \sum_{Z:q\in Z}\|h_Z\|_\infty\leq \gamma.
  \end{equation}
\end{assum}
This normalisation condition also has a natural generalisation to quasi-local interactions. Our results extend also to this case (see \cref{sec:quasi-local}):
\begin{assum}\label{quasi-local}
  The total interaction strength of the Hamiltonians touching any qudit $q$ has an exponential decay, that is:
  \begin{equation*}
    \sum_{Z:q\in {Z}}\|h_{Z}\|_\infty e^{k^{1+\varepsilon}} \leq \gamma,
  \end{equation*}
where, $\varepsilon>0$ and for each $Z$, $k=k(Z)$ is the number of qudits where $Z$ acts non-trivially.
\end{assum}
\noindent (This normalisation condition also appears in other results, such as generalised Lieb-Robinson bounds~\cite{Hastings_exp-decay}.)

\subsection{Generalised history states}\label{sec:gen_history_state}
Our results apply to a general notion of history state:
\begin{definition}[Generalised history state]\label{def:gen_history_state}
  Let $\HS = \bigotimes_{i=1}^N\HS_i$ with\linebreak $\HS_i = \left(\bigoplus_{x_i\in\Sigma_i}\HS_{x_i}\right)$, where each $\Sigma_i$ is a finite set of symbols.
  Denote $\vec{x} = x_1,x_2,\dots,x_N$.
  We define a \emph{generalised history state} $\ket{\Psi}\in\HS$ for a quantum computation on $n$ qudits for $T$ time steps to be a state of the form:
  \begin{equation}\label{eq:history-state}
    \ket{\Psi} = \sum_{p\in\cP}\alpha_p \ket{\psi_{\vec{x}(p)}},
  \end{equation}
  where $\vec{x}(p) \neq \vec{x}(p')$ for $p\neq p'$, and
  \begin{equation}\label{eq:computational}
    \ket{\psi_{\vec{x}(p)}} \in \bigotimes_{i=1}^n\HS_{\xi_i(p)} \text{ with }
    \bigl\{\xi_i(p)\bigr\}_{i=1}^n \subseteq \bigl\{x_i(p)\bigr\}_{i=1}^N.
  \end{equation}

  $\cP$ is a finite partially ordered set (poset) of cardinality $\abs{\cP} = \Ppoly(n)$, which should contain a totally ordered subset (chain) $\cT\subseteq\cP$ of cardinality $\abs{\cT}=T+1$ such that $\ket{\psi_{\vec{x}(t)}}$ is the state of the state of the quantum computation at time $t\in\cT\simeq\{0,\dots,T\}$.

  Moreover, let $t_p$ denote the maximum $t\in\cT$ such that $t_p\leq p$ (when this exists).
  Then for all $p\notin\cT$ such that $t_p$ exists, $\ket{\psi_{\vec{x}(p)}}$ should satisfy
  \begin{equation}\label{eq:junk_states}
    \ket{\psi_{\vec{x}(p)}} = V_p\ket{\psi_{\vec{x}(t_p)}},
  \end{equation}
  where $V_p$ is a unitary that satisfies one or both of the following conditions:
  \begin{enumerate}[label=(\roman*).,ref=(\roman*)]
  \item $V_p$ is independent of the quantum computation being encoded, or
  \item there is a polynomial $\Vpoly(n)$ such that $V_p$ is produced by a quantum circuit containing $\leq\Vpoly(n)$ gates.
  \end{enumerate}
\end{definition}

\noindent In most history state Hamiltonian constructions in the literature, there is a natural total order on the time labels, so we can take $\cP=\cT=\{0,1,\dots,T\}$, and the condition in \cref{eq:junk_states} does not apply.
But this is not always the case~\cite{spacetime}, and we do not impose this in our definition.

\begin{rem}\label{rem:gen_history_state}
  Letting $\Xi_i := \bigl\{\xi_i(t)\bigr\}_{t\in\cT}$, \cref{eq:history-state,eq:computational} of \cref{def:gen_history_state} together imply that
  \begin{equation}
    \forall t\;\forall x_i(t) \notin \Xi_i: \; \HS_{x_i(t)}=\C.
  \end{equation}
\end{rem}

To our knowledge, \cref{def:gen_history_state} encompasses all history state constructions in the literature.
The standard history state from \cref{eq:vanilla_history_state} for a quantum circuit on $n$ qudits, with a single $T$-dimensional qudit as its ``clock register'', is the special case of \cref{def:gen_history_state} with $N=n+1$, $\Sigma_{i\leq n}=\{q\}$, $\Sigma_N = \{0,\dots,T\}$, $\HS_q=\C^d$, $\HS_{0,\dots,T} = \C$.
Kitaev\cite{Kitaev_book} showed how to implement this with just qubits, using a unary clock construction, giving the special case $N=n+T$, $\Sigma_{i\leq n}=\{q\}$, $\Sigma_{i>n} = \{0,1\}$, $\HS_q=\C^2$, $\HS_{0,1} = \C$, $x_{i\leq n}(t) = \xi_i(t) = q$, $x_{n<i\leq n+t}(t) = 1$, $x_{i>n+t} = 0$.

However, \cref{def:gen_history_state} is substantially more general.
E.g.\ the complicated 1D construction of~\cite{AGIK} on 12-dimensional qudits\footnote{In fact, as noted in~\cite{Nagaj}, to prove QMA-hardness a 13th state needs to be added to fix a minor bug in the construction.}, where time is encoded not in a separate register but in the location of the ``active'' qubits on the chain, is a special case of \cref{def:gen_history_state} with $N=nT+O(1)$, $\Sigma_i = \{1,\dots,8\}$, $\HS_{i \leq 4}=\C$, $\HS_{i\geq 5}=\C^2$ and $\xi_i(t) = t+n \pm 1$.
(The $\pm 1$ depends on how far through a ``cycle'' $t$ is; see~\cite{AGIK} for the gory details.)

The ``space-time'' construction of~\cite{spacetime}, where each qubit has its own clock and time advances non-linearly by a process analogous to string diffusion on a torus, gives a special case of \cref{def:gen_history_state} with $N=2n$, $\Sigma_{2i-1} = \{q\}$, $\Sigma_{2i}=\{0,\dots,T\}$, $\HS_q=\C^2$, $\HS_{0,\dots,T}=\C$, $x_{2i-1}(t) = \xi_i(t) = q$, $x_{2i}(t) = t_i$ where $t=t_1,t_2,\dots,t_n \in \mathcal{T}$ runs over ``valid'' time-configurations (in the terminology of \cite{spacetime}).
In this case, the computational state $\ket{\psi_t}$ associated to $t$ is still a state obtained by applying a particular number of gates from a given quantum circuit.
But there is now only a partial order on the gates, and different $t$ correspond to applying the gates in all possible different orders consistent with this partial order.

In fact, \cref{def:gen_history_state} encompasses all history states that fit into the ``unitary labelled graph'' formalism introduced in \cite{wheelbarrow}, which itself encompasses all the history states in the literature.
\Cref{def:gen_history_state} also includes some states for which there is no known (and may not be any) corresponding local Hamiltonian construction (compare with \cite[Lemma~53]{wheelbarrow}).

\subsection{History-state amplitude distribution}\label{sec:amplitudes}
If the amplitudes $\alpha_t$ in \cref{def:gen_history_state} are extremely non-uniform, so that a very large fraction of the total amplitude is concentrated on only a tiny subset of times in the history state, then one can find trivial (and less trivial) examples of gapped Hamiltonians with such ground states.

For example, if $\alpha_0 = 1-\frac{1}{\poly(n)}$ so that almost all the amplitude is on the initial state, then regardless of the rest of the amplitude distribution, the Hamiltonian $H=\1-\proj{0}\ox\proj{\psi_0}$ is gapped and the energy of the corresponding history state is $\frac{1}{\poly(n)}$, close to the 0~ground state energy of $H$.
A less trivial example is provided by the standard Kitaev Hamiltonian $H_{\mathrm{Kitaev}}$~\cite{Kitaev_book}, but adding an energy \emph{bonus} to the initial state encoding the input to the circuit: $H = H_{\mathrm{Kitaev}} - \proj{0}\ox\proj{\psi_0}$.
This Hamiltonian has a constant spectral gap, and its ground state is a history state with exponentially decaying amplitudes $\alpha_t = O(e^{-t})$.
Note that in both cases, the Hamiltonian is $k$-local when the initial state $\psi_0$ is a product state, which is true for quantum computation in the circuit model.\footnote{For QMA verifier circuits, the witness part of the input is an arbitrary quantum state, which of course need not be product. But the witness part of the input is left unconstrained for local Hamiltonian QMA-hardness proofs, i.e.\ the Hamiltonian restricted to the time$=0$ subspace must act trivially on the witness part of the computational register. So any initial bonus or penalty terms in the Hamiltonian are tensored with $\1$ on the witness part, hence remain $k$-local.}

These examples are not interesting from a complexity theory perspective.
In the $\alpha_0 = 1-\frac{1}{\poly(n)}$ example, whilst the history state is a low-energy state of the gapped Hamiltonian $H=\1-\proj{0}\ox\proj{\psi_0}$, the ground state itself is trivial: it is the all-$\ket{0}$ state, independent of the quantum circuit being encoded.
In the case of the Kitaev Hamiltonian with bonus term on the input state, the amplitude on the output state is exponentially small.
The bonus term could instead be applied to the output state of the computation $\proj{T}\ox\proj{\psi_T}$, giving a gapped Hamiltonian with constant amplitude on the output and exponentially small amplitude on the input.
But the output of a quantum circuit is generally \emph{not} a product state.
In particular, it is almost maximally far from product for the random circuits used in the proofs of our results.
So the bonus term on the output cannot be $k$-local but must act on all $n$ qudits of the computational register, and the resulting Hamiltonian is not $k$-local (or even quasi-local).
On the other hand, other non-uniform history state constructions in the literature, such as that in \cite{Peres}, do have polynomially-decaying spectral gaps.
The Hamiltonian given in \cite{Peres} has a constant gap, but it includes interaction terms whose norm scales as $O(T)$.
Normalising appropriately (\cref{normalisation}), this gives a spectral gap that closes as $O(1/T)$.

Again, we want to choose the mildest condition on the amplitudes that rules out this kind of uninteresting example.
The condition under which our results apply is:
\begin{assum}\label{amplitudes}
  The history state amplitudes $\alpha_p$ in \cref{def:gen_history_state} should satisfy either of the following conditions:
  \begin{enumerate}
  \item\label{amplitudes:1}%
    For an increasing sequence $(n_i)_{i\in \NN}$, there exist polynomial functions of $n$, $T(n)\geq r(n)+r_1(n)$ and \[r(n), r_1(n)\geq 11050 n^2\log(d)\max\left\{ (4\Vpoly(n)d^4)^{11},\,n^{\constantC}\right\}\] such that \[\frac {\sum_{p \in \cR_2} \sum_{p' \in \cP \setminus \cR} |\alpha_p| |\alpha_{p'}|}{\sum_{p \in \cR}|\alpha_p|^2 }=O(1/n), \quad \sum_{p \in \cR}|\alpha_p|^2\geq 1/n^{\theta-1},\] where $ \cR_2:=\{p:r+1\leq t_p\leq r+r_1\}$ and $\cR=\left\{p\in \cP: p\geq r+1 \right\}$, $\constantC>0$ is a constant and $\theta\geq1$ is related to the energy approximation.
  \item\label{amplitudes:2}%
    For an increasing sequence $(n_i)_{i\in \NN}$, there exist polynomial functions of $n$, $T(n)\geq r(n)+r_1(n)$ and \[r(n), r_1(n)\geq 11050 n^2\log(d)\max\left\{ (4\Vpoly(n)d^4)^{11},\,n^{\constantC}\right\},\] $u(n)=\left\lfloor \frac{\log(T(n)/2)}{\log r(n)}\right\rfloor-1$ and $x_0(n)>1$ such that
    \[\sum_{p\in \cA_{x_0}\cup \cA_{x_0+1}}\hspace{-0.5cm}|\alpha_p|^2\leq r(n)^{-u(n)},\;\sum_{p\in\cP\setminus\cB_{x_0}}\hspace{-0.3cm}|\alpha_p|^2\geq r(n)^{-\kappa},\; \sum_{p\in\cB_{x_0}}\hspace{-0.2cm}|\alpha_p|^2 \geq r(n)^{-\kappa}\]
    where $\constantC>0$ is a constant, $\cA_1=\{p:\nexists t\in\cT, t\leq p \}\cup \{p:0\leq t_p\leq r(n)-1\}$, $\cA_x=\{p:(x-1)r(n)\leq t_p\leq xr(n)-1\}$ for $x=2,\dots,2r(n)^{u(n)}$, $\cB_x=\{p: t_p\geq xr(n)\}$ for $x=1,\dots,2r(n)^{u(n)}-1$, $\kappa=\min\{2u(n)-2,(\theta-1)\log(n)/\log(r(n))\}$ and $\theta\geq1$ is related to the energy approximation.
  \end{enumerate}
\end{assum}

Roughly speaking, the first condition says there should not be too much amplitude towards the beginning of the computation.
This includes the uniform amplitude case, as well as many non-uniform distributions.
The second condition says there should be some point during the computation with small amplitude.
This includes certain distributions with constant amplitude at the beginning.

\subsection{Main Theorem}
We can now state our main result precisely:
\begin{thm}\label{main}
  Consider any mapping from quantum circuits to Hamiltonians such that:
  \begin{enumerate}
  \item The history state of the circuit has energy within $O(1/n)$ of the ground state of the Hamiltonian.
  \item The Hamiltonian fulfils the normalisation condition of \cref{normalisation}.
  \item The history state amplitudes satisfy \cref{amplitudes}.
  \end{enumerate}
  Then for an increasing sequence of $(n_i)_{i\in\NN}$, there exist circuits of size $T(n)$ for which the spectral gap of the associated Hamiltonian is $O(1/n)$.
\end{thm}

For many history-state amplitude distributions, including e.g.\ uniform amplitudes $\alpha_i=1/(T+1)$, our proof implies the stronger result that the Hamiltonian must have a low-energy subspace with exponentially large dimension.
Our result also extends beyond $k$-local Hamiltonians, to quasi-local with exponentially-decaying interactions (see \cref{sec:quasi-local} for details).

\cref{main} is stated in terms of the circuit model, but by the Church-Turing thesis it immediately extends to other models of quantum computation, such as the Quantum Turing Machines of~\cite{Gottesman-Irani} or the Quantum Ring Machines of~\cite{wheelbarrow}.
Any quantum-computation-to-history-state-Hamiltonian mapping can be turned into a circuit-to-history-state-Hamiltonian mapping by first mapping the circuit to the computational model in question, which by the quantum Church-Turing thesis incurs at most $\poly$ overhead.

Many Hamiltonian constructions in the literature do not have history state ground states.
In particular, perturbation gadget constructions~\cite{KR,KKR,Oliveira-Terhal,Biamonte-Love,Schuch-Verstraete,Cubitt-Montanaro} give Hamiltonians whose ground state structure is not necessarily known.
However, the QMA-hardness proofs for all these constructions prove that the resulting Hamiltonian on $n$ qudits is $1/\poly(n)$-close in operator norm to a history state Hamiltonian.\footnote{The chain of $1/\poly$ approximations leading back to a history state Hamiltonian can stretch across multiple papers~\cite{Oliveira-Terhal,Biamonte-Love,Cubitt-Montanaro}.}
In particular, this implies that these Hamiltonians always have a history state whose energy is within $1/\poly(n)$ of the ground state energy, so satisfies the assumptions of \cref{main}.
Thus \cref{main} also covers all perturbation gadget constructions in the literature.

\section{Proofs of Main Results}\label{sec:proofs}
To prove our results, we cannot start from a Hamiltonian and analyse its spectral gap, as that could only prove the result for some specific Hamiltonian construction.
Instead, we must use the only things we know about our Hamiltonians: that they are local, and that there is a low-energy history state.

We first show in \cref{sec:general-to-standard} that it suffices to prove our results for the standard history states of \cref{eq:vanilla_history_state}.
We prove the result for $k$-local Hamiltonians in \cref{sec:k-local}, before extending it to exponentially-decaying, quasi-local interaction in \cref{sec:quasi-local}.

\subsection{From Generalised to Standard History States}
\label{sec:general-to-standard}
The following theorem shows that, for our purposes, a generalised history state from \cref{def:gen_history_state} can always be reduced to a standard history state from \cref{eq:vanilla_history_state}, with similar properties.
In particular, the spectral gap of the original generalised history state Hamiltonian is no larger than that of an equivalent standard history state Hamiltonian.

\begin{thm}\label{reduce_to_standard}
  Let $\HS = \bigotimes_{i=1}^N\HS_i$ as in \cref{def:gen_history_state} with $\abs{\Sigma_i}\leq D$, and let $\ket{\Psi} = \sum_p \alpha_p \ket{\psi_{\vec{x}(p)}} \in \HS$ be a generalised history state for an $n$-qudit quantum circuit.
  If there exists a Hamiltonian $H\in \cB(\HS)$ with ground state $\ket{\Psi}$ that is $k$-local with respect to the decomposition $\HS=\bigotimes_i\HS_i$, then there exists a $k'$-local Hamiltonian $H' \in \cB(\HS')$ with $\HS' = (\C^d)^{\ox N}\ox(\C^{D_i})^{\ox N}$ with $D_i\leq D$, such that:
  \begin{enumerate}[label=(\roman*).,ref=(\roman*)]
  \item
    $k' \leq 2k$,
  \item
    $H'$ has ground state
    \begin{equation}\label{eq:Psip}
      \ket{\Psi'} := \sum_{p\in\cP} \alpha_p \ket{p}\ket{\psi_{\vec{x}(p)}},
    \end{equation}
  \item\label{reduce_to_standard:clock}
    $\ket{p} \in (\C^{D_i})^{\ox N}$ are orthogonal product states,
  \item
    the spectral gap $\Delta(H) \leq \Delta(H')$.
  \end{enumerate}
\end{thm}

\begin{proof}
  Note that, by definition, $\ket{\Psi}$ has support only on the subspace\linebreak $\bigoplus_{p\in\cP}\left(\bigotimes_{i=1}^n \HS_{\xi_i(p)} \right) \subseteq \HS$  (in the notation of \cref{def:gen_history_state}).
  Since $\ket{\Psi}$ is a history state for an $n$ qudit quantum circuit, we must have $\HS_{\xi_i(p)} = \C^d$.
  By \cref{rem:gen_history_state}, $\HS_{x_i(p)\notin\Xi_i} = \C$ where $\Xi_i := \bigl\{\xi_i(p)\bigr\}_{p\in\cP}$.
  Thus we have
  \begin{equation}
    \HS_i = \biggl(\bigoplus_{x_i\in\Xi_i}\C^d\biggr)
            \oplus \biggl(\bigoplus_{x_i\in\Sigma_i\cap\Xi_i^c}\C\biggr)
          \simeq \C^d\ox\C^{\abs{\Xi_i}} \ox \C^{\abs{\Sigma_i\cap\Xi_i^c}}
  \end{equation}
  where $\C^{\abs{\Xi_i}} = \linspan\{\ket{x_i}\}_{x_i\in\Xi_i}$, and $\ket{\Psi} = \sum_p \alpha_p \ket{\psi_{\vec{x}(p)}}$ only has support on (not necessarily all of)
  \begin{equation}\label{eq:Hp}
    \HS' := \bigotimes_{i=1}^N \left(\C^d \ox \C^{\abs{\Xi_i}}\right)
         \simeq (\C^d)^{\ox N} \ox \bigotimes_{i=1}^N\C^{\abs{\Xi_i}}.
  \end{equation}

  For each local term $h_j$ of $H$, define $h'_j = h_j|_{\HS'}$ and take $H' = \sum_j h'_j = H|_{\HS'}$.
  Each $h_j$ acts non-trivially on at most $k$ subsystems $\HS_i$ of $\HS$.
  But each $\HS_i$ contains one factor $\C^d \ox \C^{\abs{\Xi_i}}$.
  Thus each $h'_j$ acts non-trivially on at most $2k$ of the tensor factors in $\HS'$.
  Hence $H'$ is $2k$-local with respect to the tensor product decomposition in \cref{eq:Hp}, as claimed.

  By \cref{def:gen_history_state}, $\ket{\psi_{\vec{x}(p)}} \in \bigotimes_{i=1}^n\HS_{\xi_i(p)} \subseteq (\C^d)^{\ox N}$.
  Thus, considered as a state in $\HS'$, $\ket{\psi_{\vec{x}(p)}}$ has the form $\ket{\psi_{\vec{x}(p)}}\ox\bigotimes_{i=1}^N\ket{x_i(p)}$.
  Since by \cref{def:gen_history_state} $\vec{x}(p) \neq \vec{x}(p')$ for $p\neq p'$, the states $\ket{p} := \bigotimes_{i=1}^N\ket{x_i(p)}$ are orthogonal product states.

  Let $\ket{\Psi'}$ be the restriction of $\ket{\Psi}$ to $\HS'$, so that $\braket{\Psi'|H|\Psi'} = \bra{\Psi} H|_{\HS'} \ket{\Psi} = \braket{\Psi'|H'|\Psi'}$.
  The minimum eigenvalue of $H$ is non-decreasing under restriction, thus $\ket{\Psi'}$ must be the ground state of $H'$.
  We can upper bound the spectral gap of $H$ by
  \begin{equation}
    \Delta(H)
    \leq \min_{\substack{\ket{\Phi}\in\HS'\\\braket{\Phi}{\Psi}=0}} \braket{\Phi|H|\Phi} - \braket{\Psi|H|\Psi}
    = \min_{\mathclap{\substack{\ket{\Phi'}\\\braket{\Phi'}{\Psi'}=0}}}
      \braket{\Phi'|H'|\Phi'} - \braket{\Psi'|H'|\Psi'}
    = \Delta(H').
  \end{equation}
\end{proof}

Thanks to \cref{reduce_to_standard}, it suffices to prove our results for standard history states of the shape of  \cref{eq:vanilla_history_state} that satisfy the condition in~\cref{reduce_to_standard}\labelcref{reduce_to_standard:clock}.
So we restrict to this case throughout the proofs of our main technical proofs.

We encapsulate for reference the constraint on the possible time encodings implied by \cref{reduce_to_standard}\labelcref{reduce_to_standard:clock}:
\begin{assum}\label{code}
  The clock states $\ket{p}$ in the standard history state of \cref{eq:vanilla_history_state} are orthogonal product states.
\end{assum}

\subsection{$k$-local Hamiltonians}\label{sec:k-local}
We now prove \cref{main} for the case of $k$-local Hamiltonians.
The proof works by constructing a low-energy state orthogonal to the ground state.
The construction of this low-energy state differs depending on which case of \cref{amplitudes} the history state amplitude distribution satisfies; the techniques for bounding the energy of this state are similar.
We prove the first case in \cref{main:k-local1}, the second in \cref{main:k-local2}.

\begin{prop}\label{main:k-local1}
  Consider any mapping from quantum circuits to Hamiltonians such that:
  \begin{enumerate}
  \item the history state of the circuit is the ground state of the Hamiltonian;
  \item the Hamiltonian fulfils the normalisation condition of \cref{normalisation};
  \item the history state satisfies \cref{code} and \cref{amplitudes}.\labelcref{amplitudes:1}.
  \end{enumerate}
  Then, for an increasing sequence $(n_i)_{i\in \NN}$, there exist circuits of size $T(n)$ for which the spectral gap of the associated Hamiltonian is $O(1/T(n))$.
\end{prop}

\noindent\textbf{Idea of the proof:}
Consider a local random circuit of size $T$ initialised in $\ket {0^n}$, and suppose, for simplicity, that its history state is given by \cref{eq:vanilla_history_state}. Initialising the circuit in an orthogonal state, $\ket{\phi_0}=\ket {0^{n-1}1}$, will give an orthogonal history state,
\[\ket {\Phi}= \sum_{t=0}^T \alpha_t \ket {t} \ket{\phi_t},\]
where $\ket{\phi_t}=U_t \ket{\phi_{t-1}}$ for $t=1,\dots,T$. After some time $r$ the resulting circuit is `close' to be random \cite{BHH} and then the states $\ket{\psi_{r+t}}$ and $\ket{\phi_{r+t}}$ are random and locally indistinguishable for every $t>0$, so each of them will contribute with the same energy in each of the two states. Thus, if one deletes the history of the system before time $r$, then the truncated versions of the ground state $\ket {\Psi}$ and the new history state $\ket {\Phi}$ will have ``essentially'' the same energy.

By \cref{normalisation} the interaction strength of the Hamiltonian terms acting non trivially in a qudit $q$ is bounded by $\gamma$, This fact will allow to delete the history of the history state before time $r$ when $T$ is sufficiently greater than $r$, without changing the energy ``too much''. In this way, denoting $\alpha=\sum_{t=0}^r \alpha_t^2$, we get an state \[ \ket {\tilde\Phi}=\frac 1 {\sqrt{1-\alpha}} \sum_{t=r+1}^T \alpha_t \ket {t} \ket{\phi_t},\] that is orthogonal to $\ket \Psi$ and whose energy is ``close'' to the ground state energy, showing that the energy gap is ``small''.

\begin{proof}[Proof of \cref{main:k-local1}]
	Consider a local random circuit over $n$ qudits of size $T$ in the following way. In each time step $t=1,...,T$ an index $i_t$ is chosen at random from $\{1,..,n-1\}$ and a unitary $U_{t}\in \mathbb U(d^2)$ is drawn at random from the Haar measure and applied to qudits $i_t$ and $i_t+1$. Note that these circuits has already been considered in \cite{BHH,Brandao,Hayden}. The circuit is initialised in $\ket{\psi_0}=\ket {0^n}$, then thanks to \cref{reduce_to_standard} (or \cref{code}) we can assume that the history state of this circuit will be given by
	\[\ket \Psi=\sum_{p \in\cP} \alpha_p \ket p \ket{\psi_p},\]
	where $\cP$ is a finite partially ordered set (poset) of cardinality $\abs{\cP} = \Ppoly(n)$, that has a totally order subset (chain) $\cT\simeq\{0,\dots,T\}$, such that $\ket{\psi_{t}}$ is the state of the quantum computation at time $t$, that is, for $t\in\cT$
	\[\ket{\psi_t}=U_t \ket{\psi_{t-1}}.\]

	Moreover, let $t_p$ denote the maximum $t\in\cT$ such that $t_p\leq p$ (when this exists).
	Then for all $p\notin\cT$ such that $t_p$ exists, $\ket{\psi_{p}}$ should satisfy
	\begin{equation*}
	\ket{\psi_{p}} = V_p\ket{\psi_{t_p}}.
	\end{equation*}

	\noindent Consider the set
        \begin{multline*}
          \mathfrak H=\Bigl\{H=\sum_Z h_Z \;:\;
          \text{ $\ket \Psi$ is the g.s.\ of $H$},\;
          h_Z \text{ is $k$-local},\\
          \forall q\in\{1,...,m\} \sum_{Z:q\in Z}\|h_Z\|_\infty\leq \gamma \Bigr\}.
          \end{multline*}
	Note that this set depends on the history state and thus on the unitary matrices. Nevertheless, we do not make this dependency explicit in the notation for simplicity.

	In order to show the result it is enough to show that for one of these local random circuits of size $T$ and any choice of $H\in \mathfrak H$, there is a state orthogonal to $\ket \Psi$, which could be dependent of $H$, whose energy differs from the energy  of $\ket \Psi$ in $O(T^{-1})$. In fact, we are going to show a stronger result, that is, there is an state $\ket{\tilde \Phi}$ that would be close in energy to $\ket \Psi$ for any Hamiltonian in $\mathfrak H$.

	Let $r\in \cT$ be a fixed integer and let $\cR$ be the set of those indices $p$ such that $t_p\in\{r+1,...,T\}$, that is, $\cR=\left\{p\in \cP: p\geq r+1 \right\}$. Let $\alpha=\sum_{p\in\cP\setminus\cR} \alpha_p^2$. Define the state \[\ket {\tilde \Phi}=\frac 1 {\sqrt{1-\alpha}}\sum_{p\in \cR} \alpha_p \ket p \ket{\phi_p},\]
	where $\ket{\phi_{r+1}}=U_{r+1}\dots U_1\ket {0^{n-1}1}$,  for $t\in \{r+2,...,T\}$, $\ket{\phi_t}=U_t \ket{\phi_{t-1}}$, and for $p\in\cR\setminus\cT$, $\ket{\phi_{p}} = V_p\ket{\phi_{t_p}}$.
	Here, $U_t$ are the circuit unitaries and $V_p$ are the same unitaries appearing in the construction of $\ket \Psi$. This state only depends on the circuit, that is, it is the same for every $H\in \mathfrak H$.
	We will show that
	\begin{equation}\label{quantitative}
	\Pr_{l.r.c.} \left\{  \max_{H\in \mathfrak H} | \bra {\Psi} H \ket {\Psi}-\bra {\tilde \Phi} H \ket {\tilde \Phi} | \geq f(T) \right\}\leq d^{-2n},
	\end{equation}
	when $n$ and $T$ are large enough and where the asymptotic behaviour of $f(T)$ is $O(T^{-1})$. That is, the probability over the set of local random circuits that there exists a Hamiltonian with the desired properties, and such that the energy of $\ket {\tilde \Phi}$ is not within $O(T^{-1})$ of the minimal energy, is exponentially small in $n$. This immediately implies the result.

	We will divide the proof of \cref{quantitative} into two technical lemmas. Define
	\[\ket {\tilde \Psi}=\frac 1 {\sqrt{1-\alpha}}\sum_{p\in\cR} \alpha_p \ket p \ket{\psi_p}.\]
	We will first prove, in \cref{lemproof1}, that the energy of $\ket {\tilde \Psi}$ is close to the energy of $\ket {\Psi}$ for any Hamiltonian in $\mathfrak  H$ and, in \cref{lemproof2}, we will show that the energy of $\ket {\Psi}$ is close to the energy of $\ket {\tilde \Phi}$ for any Hamiltonian in $\mathfrak H$. Both results will be proven with exponentially small probability over the set of local random circuits.

	\begin{lem}\label{lemproof1}
		Let $T>r,r_1>0$, then for any $0\leq \delta \leq 1/2$,
		\begin{multline*}
                  \max_{H\in \mathfrak H}| \bra {\tilde \Psi} H \ket {\tilde \Psi}-\bra {\Psi} H \ket {\Psi} |\\
                  \leq \frac {2\gamma \sum_{p \in \cR_2} \sum_{p' \in \cP \setminus \cR}  |\alpha_p| |\alpha_{p'}|+ 2\gamma \Ppoly (n) (d^{-n/2}+\delta)}{1-\alpha}
                  \end{multline*}
		with probability greater than \[1-16 \Ppoly(n)^2 \binom{n}{2}^{2\Vpoly(n)} \left(\frac{48\Vpoly(n)}{\delta}\right)^{2\Vpoly(n)d^4} \binom m k\left(\frac{12}{\delta}\right)^{d^{2k}} \left( \frac{96s_1}{d^{n}\delta^2}  \right)^{s_1/2}\] where $s_1 = \left\lfloor \left(\frac{r_1}{11050 n^2 \log (d)}\right)^{1/11}\right\rfloor$.
	\end{lem}

	\begin{proof}
		Using that $\frac 1 {1-\alpha}=1 + \frac{\alpha}{1-\alpha}$ we get
		\begin{align}
		\nonumber \mspace{20mu}&\mspace{-20mu}
                \left|\bra {\tilde \Psi} H \ket {\tilde \Psi}-\bra {\Psi} H \ket {\Psi}\right|\\
                &=\frac \alpha {1-\alpha} E_0
                  + \frac 1 {1-\alpha}  \Biggl(\sum_{p,p'\in \cR} \alpha_p^* \alpha_{p'} \bra p \bra{\psi_p} H \ket {p'} \ket{\psi_{p'}}\\
                  &\nonumber\mspace{300mu}
                    - \sum_{p,p'\in \cP} \alpha_p^*
		\alpha_{p'} \bra p \bra{\psi_p} H \ket {p'} \ket{\psi_{p'}}\Biggr)\\
		\nonumber &= \frac \alpha {1-\alpha} E_0-\frac 1 {1-\alpha}\Biggl( \sum_{p,p'\in \cP\setminus \cR} \alpha_p^* \alpha_{p'} \bra p \bra{\psi_p} H \ket {p'} \ket{\psi_{p'}}\\
                  &\nonumber\mspace{205mu}
                    + 2\mathrm{Re}{\Biggl(\sum_{p \in \cR}\sum_{ p' \in \cP \setminus \cR}\hspace{-0.1cm}\alpha_p^*\alpha_{p'}\bra p \bra{\psi_p} H  \ket {p'} \ket{\psi_{p'}}\Biggr)}\Biggr)\\
		\nonumber &\leq \frac \alpha {1-\alpha} E_0-\frac 1 {1-\alpha}\Biggl(\alpha E_0 +2\mathrm{Re}{\Biggl(\sum_{p \in \cR} \sum_{ p' \in \cP \setminus \cR}\hspace{-0.15cm} \alpha_p^* \alpha_{p'}\bra p \bra{\psi_p} H  \ket {p'} \ket{\psi_{p'}}\Biggr)} \Biggr)\\
		\nonumber &= -\frac 2 {1-\alpha} \mathrm{Re}\left( \sum_{p \in \cR} \sum_{ p' \in \cP \setminus \cR} {\alpha_p^*\alpha_{p'}\left(\bra p \bra{\psi_p} H  \ket {p'} \ket{\psi_{p'}}\right)} \right) \\
		\label{equation1}&\leq \frac 2 {1-\alpha}  {\sum_{p \in \cR} \sum_{p' \in \cP \setminus \cR} |\alpha_p|\,|\alpha_{p'}|\,\left|\bra p \bra{\psi_p} H  \ket {p'} \ket{\psi_{p'}}\right|},
		\end{align}
		where we are using that $\ket \Psi$ is the eigenvector associated to the eigenvalue of minimum energy $E_0\geq 0$.

		Now, by \cref{code}, for any $p\neq p'$,
                \begin{equation}\label{equation2}
		\left|\bra p \bra{\psi_p} H \ket {p'} \ket{\psi_{p'}}\right|\leq \gamma.
		\end{equation}
		Moreover, when $t_p\gg t_{p'}$, intuitively, we have that $\ket {\psi_{t_p}}$ is close to be Haar random and independent of $\ket {\psi_{t_{p'}}}$ and their overlap is very small even if we apply circuits not to big (of size $\Vpoly (n)$) to any of them and a $k$-local operator to one of them. This is the statement of the following Lemma whose proof will be given in the Appendix.

		\begin{lem}\label{lemmaconc1}
			Let $\mathfrak h_k=\left\{h\in B(\mathbb C^{d^m}):h=\id_{m-k} \otimes \tilde h_k,\|h\|_\infty \leq 1, h \text { is s.a.} \right\}$ be the set of $k$-local Hermitian operators on $m$ qudits with operator norm $\leq1$. For any $h\in\mathfrak h_k$, and any $p$, $p'$ such that $|t_p-t_{p'}|\geq r_1$, let  $f_h := \bra p \bra {\psi_p} h \ket {p'} \ket {\psi_{p'}}$. Then, for any $0\leq \delta \leq 1/2$,
			\begin{multline*}
                          \Pr_{l.r.c.} \left(  \max_{h\in \mathfrak h_k} |f_{h} | \geq \frac{1}{d^{n/2}}+\delta \right)\\
                          \leq 4 \binom{n}{2}^{2\Vpoly(n)} \left(\frac{48\Vpoly(n)}{\delta}\right)^{2\Vpoly(n)d^4} \binom m k\left(\frac{12}{\delta}\right)^{d^{2k}} \left( \frac{96s_1}{d^{n}\delta^2}  \right)^{s_1/2},
                          \end{multline*}
			where \[s_1 = \left\lfloor \left(\frac{r_1}{11050 n^2\log(d)}\right)^{1/11}\right\rfloor.\]
		\end{lem}
		Let $r_1$ a fixed number, take $p$, $p'$ such that $t_p-t_{p'}>r_1$ and let $q$ be one of the qudits where the encodings of $p$ and $p'$ differ. Using the fact that $\left|\bra p \bra{\psi_p} H  \ket {p'} \ket{\psi_{p'}}\right|\leq\sum_{Z:q\in Z} \left|\bra p \bra{\psi_p} h_Z  \ket {p'} \ket{\psi_{p'}}\right|$ and \cref{lemmaconc1} we have
		\begin{align}
		\nonumber&\Pr_{l.r.c.} \left\{\max_{H\in \mathfrak H}\left|\bra p \bra{\psi_p} H \ket {p'} \ket{\psi_{p'}}\right| \geq \gamma \left(\frac {1}{d^{n/2}}+\delta\right)\right\}\\
		\nonumber&\leq \Pr_{l.r.c.} \left\{\exists \{h_Z\}_Z, \sum_Z {h_Z}\in \mathfrak H : \sum_{Z:q\in Z} \left|\bra p \bra{\psi_p} h_Z  \ket {p'} \ket{\psi_{p'}}\right|
                  \geq \gamma \left(\frac {1}{d^{n/2}}+\delta\right)\right\}\\
		\nonumber&\leq \Pr_{l.r.c.} \left\{\exists \{h_Z\}_Z, \sum_Z {h_Z}\in \mathfrak H : \sum_{Z:q\in Z} \left|\bra p \bra{\psi_p} h_Z  \ket {p'} \ket{\psi_{p'}}\right|\right. \\
                  \Biggl.&\nonumber\mspace{355mu}
                    \geq \sum_{Z:q\in Z} \|h_Z\|_\infty \left(\frac {1}{d^{n/2}}+\delta\right)\Biggr\}\\
		\nonumber&\leq \Pr_{l.r.c.} \left\{\exists \{h_Z\}_Z, \sum_Z {h_Z}\in \mathfrak H : \max_{Z:q\in Z} \frac{|\bra p \bra {\psi_p} h_Z \ket {p'} \ket {\psi_{p'}}|}{\|h_Z\|_\infty}  \geq \frac {1}{d^{n/2}}+\delta \right\}\\
		\nonumber&\leq \Pr_{l.r.c.} \left\{  \max_{h\in \mathfrak h_k} |\bra p \bra {\psi_p} h \ket {p'} \ket {\psi_{p'}}|  \geq \frac {1}{d^{n/2}}+\delta \right\}\\
		&\leq 4 \binom{n}{2}^{2\Vpoly(n)} \left(\frac{48\Vpoly(n)}{\delta}\right)^{2\Vpoly(n)d^4} \binom m k\left(\frac{12}{\delta}\right)^{d^{2k}} \left( \frac{96s_1}{d^{n}\delta^2}  \right)^{s_1/2}. \label{equation3}
		\end{align}

		In order to bound, the terms of the form $|\alpha_p|\,|\alpha_{p'}|\,\left|\bra p \bra{\psi_p} H  \ket {p'} \ket{\psi_{p'}}\right|$ in \cref{equation1} we distinguish two cases, if $p\in \cR_1:=\{p:t_p\geq r+r_1\}$ then $t_p -t_{p'}\geq r_1$ and we use \cref{equation3}, and if $p\in \cR_2:=\{p:r+1\leq t_p\leq r+r_1\}$ we use \cref{equation2}. Counting the number of each of these cases and applying a union bound argument we obtain
		\begin{multline*}
		\Pr_{l.r.c.} \Biggl\{\max_{H\in \mathfrak H} \sum_{p \in \cR} \sum_{p' \in \cP \setminus \cR}  |\alpha_p|\,|\alpha_{p'}|\,\left|\bra p \bra{\psi_p} H  \ket {p'} \ket{\psi_{p'}}\right|\\
                \geq  \gamma \sum_{p \in \cR_2} \sum_{p' \in \cP \setminus \cR}  |\alpha_p| |\alpha_{p'}|+ \gamma \sum_{p\in \cR_1} \sum_{p' \in \cP \setminus \cR} |\alpha_p| |\alpha_{p'}| \left(\frac {1}{d^{n/2}}+\delta\right)\Biggr\} \\
		\leq 4 \Ppoly(n)^2 \binom{n}{2}^{2\Vpoly(n)} \left(\frac{48\Vpoly(n)}{\delta}\right)^{2\Vpoly(n)d^4} \binom m k\left(\frac{12}{\delta}\right)^{d^{2k}} \left( \frac{96s_1}{d^{n}\delta^2}  \right)^{s_1/2},
		\end{multline*}
		where $s_1 = \left\lfloor \left(\frac{r_1}{11050 n^2\log(d)}\right)^{1/11}\right\rfloor$. Relating the $l_1$ norm of the vector $\alpha=(\alpha_p)_{p\in\cP}$ with its $l_2$ norm and putting this together with \cref{equation1} we obtain the result.
	\end{proof}

	\begin{lem}\label{lemproof2}
		Let $T>r,r_1>0$, then for any $\delta>0$
		\[\max_{H\in \mathfrak H} |\bra {\tilde \Psi} H \ket {\tilde \Psi} -  \bra {\tilde \Phi} H \ket {\tilde \Phi}|\leq \left(\Ppoly(n)+m\right)\gamma\delta\]
		with probability greater than \[1-8 (\Ppoly(n)^2+m\Ppoly(n)^2)  \binom{n}{2}^{2\Vpoly(n)}\hspace{-0.1cm} \left(\frac{48\Vpoly(n)}{\delta}\right)^{2\Vpoly(n)d^4}\hspace{-0.1cm}\binom m k \left(\frac{12}{\delta}\right)^{d^{2k}}\hspace{-0.1cm} \left( \frac{6144s}{d^{n}\delta^2}\right)^{s/2},\]
		where $s= \left\lfloor \left(\frac{r}{1400 n^2\log(d)}\right)^{1/11}\right\rfloor$.
	\end{lem}
	In order to prove this Lemma we are going to need the following Lemma which will be proven in the Appendix.

	\begin{lem}\label{lemmaconc2}
		Let $\mathfrak h_k=\left\{h\in B(\mathbb C^{d^m}):h=\id_{m-k} \otimes \tilde h_k,\|h\|_\infty \leq 1, h \text { is s.a.} \right\}$ be the set of $k$-local Hermitian operators on $m$ qudits with operator norm $\leq1$. For any $h\in\mathfrak h_k$, let $f_h := \bra p \bra {\psi_p} h \ket {p'} \ket {\psi_{p'}} -  \bra p\bra {\phi_p} h\ket {p'} \ket {\phi_{p'}}$ where $p,p'\geq r$. 
		Then,
		\[\Pr_{l.r.c.} \left(\max_{h\in\mathfrak h_k} |f_h| \geq \delta  \right)
		\leq 8 \binom{n}{2}^{2\Vpoly(n)}\hspace{-0.1cm} \left(\frac{48\Vpoly(n)}{\delta}\right)^{2\Vpoly(n)d^4}\hspace{-0.15cm} \binom m k \left(\frac{12}{\delta}\right)^{d^{2k}}\hspace{-0.15cm} \left( \frac{6144 s}{d^{n}\delta^2}  \right)^{s/2}\hspace{-0.15cm},\]
		where
		\[s = \left\lfloor \left(\frac{r}{1400 n^2\log(d)}\right)^{1/11}\right\rfloor.\]
	\end{lem}

	\begin{proof}[Proof of \cref{lemproof2}]
		We have that
		\begin{multline*}
                  |\bra {\tilde \Psi} H \ket {\tilde \Psi} -  \bra {\tilde \Phi} H \ket {\tilde \Phi}|\\
                  \leq\sum_{p,p' \in\cR} |\alpha_p|\,|\alpha_{p'}|\,|\bra p \bra{\psi_p} H \ket {p'} \ket{\psi_{p'}}-\bra p \bra{\phi_p} H \ket {p'} \ket{\phi_{p'}}|,
                \end{multline*}
		For fixed $p\neq p'$, we have that $\bra p \bra{\psi_p} h_Z \ket {p'} \ket{\psi_{p'}}=\bra p \bra{\phi_p} h_Z \ket {p'} \ket{\phi_{p'}}=0$ unless $h_Z$ acts non trivially in the qudits where the encodings of $p$ and $p'$ differ. Let $q$ be one of these qudits, we have that
		\begin{multline*}
                  \left|\bra p \bra{\psi_p} H  \ket {p'} \ket{\psi_{p'}}-\bra p \bra{\phi_p} H \ket {p'} \ket{\phi_{p'}}\right|\\
                  \leq \sum_{Z:q\in Z} \left|\bra p \bra{\psi_p} h_Z  \ket {p'}\ket{\psi_{p'}}-\bra p \bra{\phi_p} h_Z \ket {p'} \ket{\phi_{p'}}\right|
                \end{multline*}

		Now, using \cref{lemmaconc2} we have
		\begin{align}
                  \mspace{20mu}&\mspace{-20mu}
		\nonumber\Pr_{l.r.c.} \left\{\max_{H \in \mathfrak H}\left|\bra p \bra{\psi_p} H \ket {p'} \ket{\psi_{p'}}-\bra p \bra{\phi_p} H \ket {p'} \ket{\phi_{p'}}\right|\geq \gamma\delta\right\} \nonumber\\
		&\nonumber\leq \Pr_{l.r.c.} \Biggl\{\exists \{h_Z\}_Z, \sum_Z {h_Z}\in \mathfrak H : \\
                &\nonumber\mspace{135mu}
               \sum_{Z:q\in Z} \left|\bra p \bra{\psi_p} h_Z  \ket {p'} \ket{\psi_{p'}}-\bra p \bra{\phi_p} h_Z  \ket {p'} \ket{\phi_{p'}}\right|\geq \gamma\delta \Biggr\}\\
		&\nonumber\leq \Pr_{l.r.c.} \Biggl\{\exists \{h_Z\}_Z, \sum_Z {h_Z}\in \mathfrak H : \\
                &\nonumber\mspace{50mu}
                  \sum_{Z:q\in Z} \left|\bra p \bra{\psi_p} h_Z  \ket {p'} \ket{\psi_{p'}}-\bra p \bra{\phi_p} h_Z  \ket {p'} \ket{\phi_{p'}}\right|\geq \sum_{Z:q\in Z} \|h_Z\|_\infty \delta \Biggr\}\\
		&\nonumber\leq \Pr_{l.r.c.} \Biggl\{\exists \{h_Z\}_Z, \sum_Z {h_Z}\in \mathfrak H:  \\
                &\nonumber\mspace{140mu}
                  \max_{Z:q\in Z} \frac{\left|\bra p \bra {\psi_p} h_Z \ket {p'} \ket {\psi_{p'}}-\bra p \bra {\phi_p} h_Z \ket {p'} \ket {\phi_{p'}}\right|}{\|h_Z\|_\infty} \geq \delta \Biggr\}\\
		&\nonumber\leq \Pr_{l.r.c.} \left\{  \max_{h\in \mathfrak h_k} \left|\bra p \bra {\psi_p} h \ket {p'} \ket {\psi_{p'}}-\bra p \bra {\phi_p} h \ket {p'} \ket {\phi_{p'}}\right| \geq \delta \right\}\\
		&\leq 8 \binom{n}{2}^{2\Vpoly(n)} \left(\frac{48\Vpoly(n)}{\delta}\right)^{2\Vpoly(n)d^4} \binom m k \left(\frac{12}{\delta}\right)^{d^{2k}} \left( \frac{6144 s}{d^{n}\delta^2}  \right)^{s/2}. \label{equation4}
		\end{align}
		For a fixed $p=p'>r$, we have that
		\begin{multline*}
                  \left|\bra p \bra{\psi_p} H  \ket p \ket{\psi_p}-\bra p \bra{\phi_p} H \ket p \ket{\phi_p}\right|\\
                  \leq \sum_{q=1}^m \sum_{Z:q\in Z}  \left|\bra p \bra{\psi_p} h_Z  \ket p \ket{\psi_p}-\bra p \bra{\phi_p} h_Z \ket p \ket{\phi_p}\right|.
                \end{multline*}
		Reasoning as before, for a particular $q$, we get
		\begin{multline*}\Pr_{l.r.c.} \left\{\max_{\sum_Z h_Z\in \mathfrak H}\sum_{Z:q\in Z} \left|\bra p \bra{\psi_p} h_Z  \ket p \ket{\psi_p}-\bra p \bra{\phi_p} h_Z  \ket p \ket{\phi_p}\right|\geq \gamma \delta \right\}\\
		\leq 8 \binom{n}{2}^{2\Vpoly(n)} \left(\frac{48\Vpoly(n)}{\delta}\right)^{2\Vpoly(n)d^4} \binom m k \left(\frac{12}{\delta}\right)^{d^{2k}} \left( \frac{6144 s}{d^{n}\delta^2}  \right)^{s/2}.
                \end{multline*}
		Counting the terms of each kind and applying a union bound we get
		\begin{align*}
                  &\Pr_{l.r.c.}\left\{\max_{H \in \mathfrak H}|\bra {\tilde \Psi} H \ket {\tilde \Psi} -  \bra {\tilde \Phi} H \ket {\tilde \Phi}|\geq \left(\sum_{p, p'\in\cR} |\alpha_p||\alpha_{p'}|+ m\sum_{p\in\cR} |\alpha_p|^2 \right)\gamma \delta\right\}\\
		&\mspace{20mu}\leq 8 (m+1)\Ppoly(n)^2  \binom{n}{2}^{2\Vpoly(n)} \left(\frac{48\Vpoly(n)}{\delta}\right)^{2\Vpoly(n)d^4}
                  \binom m k \left(\frac{12}{\delta}\right)^{d^{2k}} \left( \frac{6144s}{d^{n}\delta^2}\right)^{s/2}.
                \end{align*}
		Relating the $l_1$ norm of the vector $\alpha=(\alpha_p)_{p=r+1}^T$ with its $l_2$ norm finishes the proof of \cref{lemproof2}.
	\end{proof}

	Putting together \cref{lemproof1,lemproof2} we get that for any $\delta>0$
	\begin{align*}
          &\Pr_{l.r.c.}\Biggl\{\max_{H \in \mathfrak H} |\bra {\Psi} H \ket {\Psi} -  \bra {\tilde \Phi} H \ket {\tilde \Phi}|\\
          &\mspace{60mu}
            \geq \frac {2\gamma \sum\limits_{p \in \cR_2} \sum\limits_{p' \in \cP \setminus \cR}  |\alpha_p| |\alpha_{p'}|+ 2\gamma \Ppoly (n) (d^{-n/2}+\delta)}{1-\alpha}+\gamma\left(\Ppoly(n)+m\right)\delta \Biggr\}\\
          &\leq \left(16 \Ppoly(n)^2 \left( \frac{96s_1}{d^{n}\delta^2}  \right)^{s_1/2}+8 (m+1)\Ppoly(n)^2 \left( \frac{6144s}{d^{n}\delta^2}\right)^{s/2}\right)\times\\ &\mspace{200mu}
          \binom{n}{2}^{2\Vpoly(n)} \left(\frac{48\Vpoly(n)}{\delta}\right)^{2\Vpoly(n)d^4} \binom m k\left(\frac{12}{\delta}\right)^{d^{2k}}.
        \end{align*}
Taking $\delta=\frac {(1-\alpha)} {(\Ppoly(n)+m)T}$, $r=r_1$, then $s_1\leq s$ and we get
\begin{align*}
\Pr_{l.r.c.}&\Biggr\{\max_{H \in \mathfrak H} |\bra {\Psi} H \ket {\Psi} -  \bra {\tilde \Phi} H \ket {\tilde \Phi}|\\
 &\mspace{80mu}\geq  \frac {2\gamma \sum_{p \in \cR_2} \sum_{p' \in \cP \setminus \cR}  |\alpha_p| |\alpha_{p'}|+ 2\gamma \Ppoly (n)d^{-n/2}}{1-\alpha}+ \frac {3\gamma} {T} \Biggl\}\\
&\leq \left(24 \Ppoly(n)^2 +m\Ppoly(n)^2\right)\left( \frac{6144s(\Ppoly(n)+m)^2T^2}{(1-\alpha)^2d^{n}}\right)^{s/2}\binom{n}{2}^{2\Vpoly(n)}\times\\
&\mspace{80mu} \left(\frac{48\Vpoly(n)(\Ppoly(n)+m)T}{1-\alpha}\right)^{2\Vpoly(n)d^4} \binom m k \left(\frac{12(\Ppoly(n)+m)T}{1-\alpha}\right)^{d^{2k}},
\end{align*}
for a sufficiently large $n$.
Now, in order to show the result it is enough to show that for any finite $k$ and $\gamma$ and for $m$ any polynomial function of $n$, there is an election of parameters such that the previous probability is smaller than 1.

Recall that $T, m,\Vpoly(n),\Ppoly(n), \frac 1 {1-\alpha}$ are all polynomial functions of $n$. If $k \leq \frac{\log n}{2\log d}$, then by \cref{amplitudes}.\labelcref{amplitudes:1} $r(n)\geq  11050 n^2\log(d)\max\{ (4\Vpoly(n)d^4)^{11},n^{\constantC}\}$, and we have that $s=s(n)\geq 4\Vpoly(n)d^4$ and $d^{n s }\geq  n^{\beta d^{2k}}$ for any fixed $\beta$ for sufficiently large $n$. Thus
\[\Pr_{l.r.c.}\left\{\max_{H \in \mathfrak H} |\bra {\Psi} H \ket {\Psi} -  \bra {\tilde \Phi} H \ket {\tilde \Phi}|
\geq  O(1/T) \right\}\leq O(d^{-n^{2}/2})\leq O(d^{-2n}),\]
for a sufficiently large $n$.
This proves \cref{quantitative} which finishes the proof.
\end{proof}

\begin{rem}
	Although \cref{main:k-local1} was stated for Hamiltonians that are $k$-local and this entails implicitly that $k$ is finite, note that the result holds for $k$-local Hamiltonians with $k\leq C \log n$ and $C$ being a constant. This is the only assumption made on $k$ at the end of the proof.
\end{rem}

\begin{rem}\label{Rem2}
	In the proof of \cref{main:k-local1} the initialisation to $\ket {0^{n-1}1}$ of the state $\ket {\tilde \Phi}$ is not important. The same proof holds for any history state initialised in any product state and such that the history of the first $r$ times is deleted. Thus, any of these ``truncated history states'' have energy close to the minimal energy with high probability. Moreover, applying a union bound over all this events we can show that the probability of having a subspace of dimension $d^n$ of energy within $O(T^{-1})$ of the ground states is $1-d^nd^{-2n}=1-d^{-n}$.
\end{rem}

\begin{rem}\label{Rem3}
	The result of \cref{main:k-local1} extends straightforwardly to the case where the history state $\ket \Psi$ is not the ground state but its energy is $O(1/n^\theta)$ close to the minimal energy, that is, $\bra \Psi H\ket \Psi=E_0+O(1/n^\theta)$. In this case the spectral gap will be $O(1/n)$. Indeed, the only place in the proof where we are using that $\ket \Psi$ is the ground state is in the proof of \cref{lemproof1} to bound the energy and the same argument can still be applied with an extra term of (by \cref{amplitudes}.\labelcref{amplitudes:1}) $O(1/n^\theta)/(1-\alpha)=O(1/n)$. Note that at the beginning of the proof of \cref{lemproof1} we are also using that the energy of $\ket \Psi$ does not decrease after deleting part of the history of the circuit. This could happen now, but in this case \cref{lemproof1} is not needed as the decrease in energy will give us a better upper bound for the difference between $\bra {\tilde \Phi} H \ket {\tilde \Phi}$ and the ground state energy, which is the final goal.
\end{rem}

\begin{prop}\label{main:k-local2}
	Consider any mapping from quantum circuits to Hamiltonians such that:
	\begin{enumerate}
		\item the history state of the circuit is the ground state of the Hamiltonian;
		\item the Hamiltonian fulfils the normalisation condition of \cref{normalisation}.
		\item the history state satisfies \cref{code} and \cref{amplitudes}.\labelcref{amplitudes:2}.
	\end{enumerate}
	Then, for an increasing sequence $(n_i)_{i\in \NN}$, there exist circuits of size $T(n)$ for which the spectral gap of the associated Hamiltonian is $O(1/\poly(n))$.
\end{prop}

\begin{proof}
  Consider a local random circuit over $n$ qudits of size $T$, thanks to \cref{reduce_to_standard} (or \cref{code}) we can assume that the history state of this circuit will be given by
	\[\ket \Psi=\sum_{p \in\cP} \alpha_p \ket p \ket{\psi_p},\]

	Let $H$ be a $k$-local Hamiltonian such that $\ket {\Psi}=\sum_{p\in\cP} \alpha_p \ket p \ket{\psi_p}$ is its ground state. By \cref{amplitudes}.\labelcref{amplitudes:2} there exist a number $x_0$ such that $\sum_{p\in \cA_{x_0}\cup \cA_{x_0+1}}|\alpha_p|^2\leq r^{-u}$, $\lambda_{x_0}:=\sum_{p\in\cP\setminus\cB_{x_0}}|\alpha_p|^2\geq r^{-2u+2}$ and $1-\lambda_{x_0}=\sum_{p\in\cB_{x_0}}|\alpha_p|^2 \geq r^{-2u+2}$ 	where $\cA_1=\{p:\nexists t\in\cT, t\leq p \}\cup \{p:0\leq t_p\leq r-1\}$, $\cA_x=\{p:(x-1)r\leq t_p\leq xr-1\}$ for $x=2,\dots,2r^u$, $\cB_x=\{p: t_p\geq xr\}$ for $x=1,\dots,r^u-1$. With this notation,

	\[\ket {\Psi}=\sqrt \lambda \ket {\xi_0} + \sqrt {1-\lambda} \ket {\xi_1}.\]
	We will show that for any $k$-local Hamiltonian such that $\ket {\Psi}$ is the ground state, then  both $D_0$ the energy of $\ket{\xi_0}$ and $D_1$ the energy of $\ket{\xi_1}$ are $1/\poly(n)$ close to the ground state energy. As $\ket{\xi_0}$ and $\ket{\xi_1}$ are orthogonal, this will automatically imply the result.

	The ground state energy is
	\begin{align}
	\nonumber E_0 &=\bra {\Psi}H\ket {\Psi}\\
	\nonumber &=\lambda \bra {\xi_0}H\ket {\xi_0}+(1-\lambda) \bra {\xi_1}H\ket {\xi_1}+2\lambda(1-\lambda)\mathrm{Re}\left(\bra{\xi_0}H\ket{\xi_1}\right)\\
	\label{gsenergy}&=\lambda D_0+(1-\lambda)D_1 +2\lambda (1-\lambda)\mathrm{Re}\left(\bra{\xi_0}H\ket{\xi_1}\right)\\
	\nonumber &=\lambda (D_0-E_0+E_0)+(1-\lambda)(D_1-E_0+E_0) +2\lambda (1-\lambda)\mathrm{Re}\left(\bra{\xi_0}H\ket{\xi_1}\right)
	\end{align}

	Reordering the terms we get
        \begin{align*}
          \lambda(D_0&-E_0)+(1-\lambda)(D_1-E_0)\\ &=-2\mathrm{Re}\left(\sum_{p\in\cP\setminus\cB_{x_0}}\sum_{p'\in\cB_{x_0}} \alpha_p^* \alpha_{p'} \bra p \bra{\psi_p}H \ket {p'} \ket{\psi_{p'}}\right)\\
	  &\leq 2\sum_{p\in\cP\setminus\cB_{x_0}}\sum_{p'\in\cB_{x_0}}  \left|\alpha_p^* \alpha_{p'} \bra p \bra{\psi_p}H \ket {p'} \ket{\psi_{p'}}\right|.
	\end{align*}
	Repeating the analysis carried out in \cref{lemproof1} the terms of the form $\left|\bra p \bra{\psi_p}H \ket {p'} \ket{\psi_{p'}}\right|$ can be bounded in two different ways, those where $p$ and $p'$ are inside the slice $x_0$ and those where at least one of them is not in the slice $x_0$ and so $|t_p-t_{p'}|>r$. Let $\cF$ denote the set of those $(p,p')$ of the latter case, that is,
	\[\cF=\left\{(p,p'): p\in\cP\setminus\cB_{x_0},p'\in\cB_{x_0}, |t_p-t_p'|>r\right\}.\]
	The terms where $(p,p')\notin\cF$ will be bounded by \cref{equation2}. Thus,
	\begin{align*}\lambda(D_0&-E_0)+(1-\lambda)(D_1-E_0)\\
         &\leq 2\gamma \sum_{p\in \cA_{x_0}}  |\alpha_p| \sum_{p\in \cA_{x_0+1}}|\alpha_{p'}|+2 \sum_{(p,p')\in \cF}  \left|\alpha_p^* \alpha_{p'} \bra p \bra{\psi_p}H \ket {p'} \ket{\psi_{p'}}\right|\\
	&\leq2\gamma r^{-2u+1}+2 \sum_{(p,p')\in \cF}  \left|\alpha_p^* \alpha_{p'} \bra p \bra{\psi_p}H \ket {p'} \ket{\psi_{p'}}\right|.\end{align*}
	The terms where $(p,p')\in\cF$ can be bounded by \cref{equation3}. Overcounting the number of these terms, using the inequalities between the $l_1$ and $l_2$ norms of vectors, and applying a union bound we get that, for any $\delta>0$,
	\begin{align*}
	\Pr_{l.r.c.}&\biggl\{ \max_{H\in \mathfrak H} \{\lambda(D_0-E_0)+(1-\lambda)(D_1-E_0)\}\\
	&\mspace{130mu}\geq \gamma \left(r^{-2u+1}+\Ppoly(n)d^{-n/2}+\Ppoly(n)\delta\right)\biggr\}\\
	&\leq \Pr_{l.r.c.}\left\{ \max_{H\in \mathfrak H}\sum_{(p,p')\in \cF}  \left|\alpha_p^* \alpha_{p'} \bra p \bra{\psi_p}H \ket {p'} \ket{\psi_{p'}}\right|\geq \gamma \Ppoly(n) \left(d^{-n/2}+\delta\right)\right\}\\
	&\leq\Pr_{l.r.c.}\Biggl\{ \max_{H\in \mathfrak H}\sum_{(p,p')\in \cF}  \left|\alpha_p^* \alpha_{p'} \bra p \bra{\psi_p}H \ket {p'} \ket{\psi_{p'}}\right|\\
	&\mspace{130mu} \geq 2\gamma \sum_{(p,p')\in \cF} |\alpha_p| |\alpha_{p'}| \left(d^{-n/2}+\delta\right)\Biggr\}\\
	&\leq\Pr_{l.r.c.}\Biggl\{\exists  (p,p')\in \cF, \,  \max_{H\in \mathfrak H}  \left|\alpha_p^* \alpha_{p'} \bra p \bra{\psi_p}H \ket {p'} \ket{\psi_{p'}}\right|\\
	&\mspace{130mu} \geq\gamma |\alpha_p| |\alpha_{p'}| \left(d^{-n/2}+\delta\right)\Biggr\}\\
	&\leq\sum_{(p,p')\in\cF}\Pr_{l.r.c.}\left\{\max_{H\in \mathfrak H}  \left| \bra p \bra{\psi_p}H \ket {p'} \ket{\psi_{p'}}\right|\geq\gamma \left(d^{-n/2}+\delta\right)\right\}\\
	&\leq 4\Ppoly(n)^2 \binom{n}{2}^{2\Vpoly(n)} \left(\frac{48\Vpoly(n)}{\delta}\right)^{2\Vpoly(n)d^4} \binom m k\left(\frac{12}{\delta}\right)^{d^{2k}} \left( \frac{96s_1}{d^{n}\delta^2}  \right)^{s_1/2},
	\end{align*}
	where $s_1 = \left\lfloor \left(\frac{r_1}{11050 n^2\log(d)}\right)^{1/11}\right\rfloor$.

	Take $\delta(n)=r^{-2u+1}/\Ppoly(n)$. Recall that $T, m,\Vpoly(n),\Ppoly(n)$  are all polynomial functions of $n$. If $k \leq \frac{\log n}{2\log d}$, then by \cref{amplitudes}.\labelcref{amplitudes:2} \[r_1(n)\geq  11050 n^2\log(d)\max\{ (4\Vpoly(n)d^4)^{11},n^{\constantC}\},\] and we have that $s_1(n)\geq 4\Vpoly(n)d^4$ and $d^{n s_1}\geq  n^{\beta d^{2k}}$ for any fixed $\beta$ and sufficiently large $n$.
	Thus
	\begin{equation}\label{eq:k-local2}\Pr_{l.r.c.}\left\{  \max_{H\in \mathfrak H} \{\lambda(D_0-E_0)+(1-\lambda)(D_1-E_0)\}\geq O (r^{-2u+1}) \right\}\leq O(d^{-2n}),\end{equation}
	for a sufficiently large $n$.

	Now, \cref{amplitudes}.\labelcref{amplitudes:2} implies $\lambda,1-\lambda\geq r^{-2u+2}$ and we get $D_0-E_0+D_1-E_0\leq O (r^{-1})$, as we wanted to show.
\end{proof}

\begin{rem}\label{Rem1a}
	Again, the result of \cref{main:k-local2} extends to the case where the history state $\ket \Psi$ is not the ground state but its energy is $O(n^{-\theta})$ close to the minimal energy.
	The only place in the proof where we are using that $\ket \Psi$ is the ground state is in \cref{gsenergy} when considering the energy of in the proof of \cref{lemproof1} to bound the energy $\ket \Psi$, adding an extra term $O(n^{-\theta})$ to the discussion. Thus, \cref{eq:k-local2} is still true adding $O(n^{-\theta})$ to the lower bound inside the probability. Finally, \cref{amplitudes}.\labelcref{amplitudes:2} implies $\lambda,1-\lambda\geq r^{-\kappa}\geq n^{-\theta+1}$. Thus, $D_0-E_0+D_1-E_0\leq O (n^{-1})$ and the result follows.\end{rem}

\section{Quasi-local interactions}\label{sec:quasi-local}
Almost all history state constructions in the literature are for $k$-local Hamiltonians, i.e.\ Hamiltonians with strictly local interactions acting on at most $k$ qudits.
However, our results also extend to the case of quasi-local Hamiltonians with exponentially decaying, i.e.\ Hamiltonians with local interactions that can act on arbitrarily many qudits at once, but whose strength decays sufficiently fast with the number of qudits involved in the interaction.
In this section, we extend all our results to this setting.

\begin{thm}\label{main:quasi-local}
	Consider any mapping from quantum circuits to Hamiltonians such that:
	\begin{enumerate}
		\item the history state of the circuit is the ground state of the Hamiltonian;
		\item the Hamiltonian fulfils the normalisation condition of \cref{quasi-local};
		\item the history state satisfies \cref{code} and \cref{amplitudes}.\labelcref{amplitudes:1}.
\end{enumerate}
  Then, for an increasing sequence $(n_i)_{i\in \NN}$, there exist circuits of size $T(n)$ for which the spectral gap of the associated Hamiltonian is $O(1/T(n))$.
\end{thm}

\begin{proof}
  The proof follows the same steps as the proof of \cref{main:k-local1} until the chain of inequalities in \cref{equation3}. A similar chain of inequalities can be derive using the fact that $\left|\bra p \bra{\psi_p} H  \ket {p'} \ket{\psi_{p'}}\right|\leq\sum_{Z:q\in Z} \left|\bra p \bra{\psi_p} h_Z  \ket {p'} \ket{\psi_{p'}}\right|$ and \cref{lemmaconc1}. That is,
  \begingroup
  \allowdisplaybreaks
  \begin{align}
	\nonumber&\Pr_{l.r.c.} \left\{\max_{H\in\mathfrak H} \left|\bra p \bra{\psi_p} H \ket {p'} \ket{\psi_{p'}}\right|\geq \gamma \left(\frac {1}{d^{n/2}}+\delta\right)\right\}\\
	&\nonumber\leq \Pr_{l.r.c.} \left\{\exists \{h_Z\}_Z, \sum_Z {h_Z}\in \mathfrak H :\sum_{Z:q\in Z} \left|\bra p \bra{\psi_p} h_Z  \ket {p'} \ket{\psi_{p'}}\right|\geq \gamma \left(\frac {1}{d^{n/2}}+\delta\right)\right\}\\
	&\nonumber\leq \Pr_{l.r.c.} \left\{\exists \{h_Z\}_Z, \sum_Z {h_Z}\in \mathfrak H : \right.\\
	&\nonumber\mspace{150mu}
	\sum_k\sum_{Z:q\in Z, |Z|=k} \left|\bra p \bra{\psi_p} h_{Z}  \ket {p'} \ket{\psi_{p'}}\right| \\
	&\nonumber\mspace{200mu}
	\left.\geq \sum_k \sum_{Z:q\in Z, |Z|=k} \|h_{Z}\|_\infty e^{k^{1+\varepsilon}}  \left(\frac {1}{d^{n/2}}+\delta\right)\right\} \\
	&\nonumber\leq \Pr_{l.r.c.} \left\{\exists \{h_Z\}_Z, \sum_Z {h_Z}\in \mathfrak H: \right.\\
	&\nonumber\mspace{150mu}
	\exists k:\sum_{Z:q\in Z, |Z|=k} \left|\bra p \bra{\psi_p} h_{Z}  \ket {p'} \ket{\psi_{p'}}\right| \\
	&\nonumber\mspace{200mu}
	\left.\geq \sum_{Z:q\in Z, |Z|=k} \|h_{Z}\|_\infty e^{k^{1+\varepsilon}}  \left(\frac {1}{d^{n/2}}+\delta\right)\right\}\\
	&\nonumber\leq\sum_{k=1}^m \Pr_{l.r.c.} \left\{\exists \{h_Z\}_Z, \sum_Z {h_Z}\in \mathfrak H: \right.\\
	&\nonumber\mspace{150mu}
	\left.  \max_{Z:q\in Z, |Z|=k} \frac{|\bra p \bra {\psi_p} h_{Z} \ket {p'} \ket {\psi_{p'}}|}{\|h_{Z}\|_\infty}
	\geq \frac {1}{d^{n/2}}+e^{ k^{1+\varepsilon}}\delta \right\}\\
	&\nonumber\leq \sum_{k=1}^m\Pr_{l.r.c.} \left\{  \max_{h\in \mathfrak h_k} |\bra p \bra {\psi_p} h \ket {p'} \ket {\psi_{p'}}|  \geq \frac {1}{d^{n/2}}+e^{ k^{1+\varepsilon}}\delta \right\}\\
	&\leq 4 \binom n 2^{\Vpoly(n)} \sum_{k=1}^m  \binom m k \left(\frac{48\Vpoly(n)}{ e^{k^{1+\varepsilon}}\delta}\right)^{2\Vpoly(n)d^4} \left(\frac{12}{e^{ k^{1+\varepsilon}}\delta}\right)^{d^{2k}} \left( \frac{96s_1}{d^{n}e^{2k^{1+\varepsilon}}\delta^2}\right)^{s_1/2}, \label{equation5}
  \end{align}
  \endgroup
  where in the second inequality we are using \cref{quasi-local} and the fact that we can sum over all $Z$ such that $q\in Z$ in two steps, first summing over the $Z$ such that $q\in Z$ and $|Z|=k$ and then summing over $k$. The other inequalities follow from basic facts such as: if we have two sums with the same number of elements, we can pair the terms of these sums so that for a sum to be bigger than the other, at least one of the terms of the sum must be bigger than its pair.

	Following the same reasoning as in the proof of \cref{main:k-local1}, and making an analogous reasoning to the one in \cref{equation5} for the chain of inequalities in \cref{equation4}, we get (taking, $\delta=\frac {(1-\alpha)} {(\Ppoly(n)+m)T}$, $r=r_1$ then $s_1\leq s$)
	\begin{align*}
	p:&=\Pr_{l.r.c.}\Biggr\{\max_{H \in \mathfrak H} |\bra {\Psi} H \ket {\Psi} -  \bra {\tilde \Phi} H \ket {\tilde \Phi}|\\
	&\mspace{80mu}\geq  \frac {2\gamma \sum_{p \in \cR_2} \sum_{p' \in \cP \setminus \cR}  |\alpha_p| |\alpha_{p'}|+ 2\gamma \Ppoly (n)d^{-n/2}}{1-\alpha}+ \frac {3\gamma} {T} \Biggl\}\\
	&\leq \left(24 \Ppoly(n)^2 +m\Ppoly(n)^2\right)   \binom{n}{2}^{2\Vpoly(n)} \times\\
	&\nonumber\mspace{80mu}
	\sum_{k=1}^m  \binom m k \left( \frac{6144s}{d^{n}\delta^2 e^{2k^{1+\varepsilon}}}\right)^{s/2}\left(\frac{48\Vpoly(n)}{e^{k^{1+\varepsilon}}\delta}\right)^{2\Vpoly(n)d^4}  \left(\frac{12}{e^{k^{1+\varepsilon}}\delta}\right)^{d^{2k}}\hspace{-0.2cm},
	\end{align*}
	for a sufficiently large $n$. Taking $\delta=n^{-\log(n)^{\varepsilon}}=e^{-\log(n)^{1+\varepsilon}}$ and $r=11050n^{13}\log d $ we have that $s\geq n$. Take $T\gg r$ a polynomial function of $n$ as big as desired. Recall that $m$ is a polynomial function of $n$.
	Setting $$p_k=\binom m k \left( \frac{6144s}{d^{n}\delta^2 e^{2k^{1+\varepsilon}}}\right)^{s/2}\left(\frac{48\Vpoly(n)}{e^{k^{1+\varepsilon}}\delta}\right)^{2\Vpoly(n)d^4}  \left(\frac{12}{e^{k^{1+\varepsilon}}\delta}\right)^{d^{2k}},$$ we have that:
	\begin{itemize}
		\item If $k\leq \frac{\log n}{2\log d}$, then the reasoning at the end of \cref{main:k-local1} shows $p_k\leq O(d^{-2n})$.
		\item If $k\geq \frac{\log n}{2\log d}$, then $\left(\frac{12}{e^{k^{1+\varepsilon}}\delta}\right)^{d^{2k}}\leq 1$, for sufficiently large $n$ and $p_k\leq O(d^{-2n})$;
	\end{itemize}

	Thus, $p\leq O(d^{-2n})$ and there exist a circuit of size $T$ such that its gap is $O(T^{-1})$.
\end{proof}

Note that \cref{Rem2,Rem3} still hold. That is, we have proven that with probability $1- d^{-n}$ there is a subspace of dimension $d^n$ which has energy $O(T^{-1})$ for $n$ and $T$ sufficiently large. Moreover, if the history state energy is not the minimal energy but it is $1/\poly (n)$ close to it then there is a subspace of dimension $d^n$ which has energy $O(1/\poly(n))$.

\section{Conclusions}
In the classical complexity literature, the idea that hard instances of SAT problems occur near phase transitions has been extensively explored in the context of random-SAT problems\cite{Selman}, though proving this rigorously is challenging\cite{Naor}.
In the quantum setting, a paper by Brandao and Harrow\cite{Brandao-Harrow} proves (amongst many other results) that an analogous intuition is true in a different context: any quantum circuit that achieves a computational speed-up over classical algorithms must at some point during the computation produce a quantum state that is ``critical'', in the sense that it has long-range correlations.
In the context of adiabatic quantum computation, Gant and Somma~\cite{Somma} used query complexity bounds to prove that the spectral gap along an adiabatic path must close as $O(1/n)$ for Feynman-Kitaev-style Hamiltonian constructions.
Our results lend more support to the intuition that computational complexity is related to criticality of Hamiltonians.

Similar results to ours can be proven for standard history states (\cref{eq:vanilla_history_state}) using techniques inspired by classical Markov chain theory~\cite{Crosson}.
The two proof approaches are very different -- purely combinatorial rather than applying the probabilistic method -- and consequently have complementary strengths and weaknesses.
The history states we show here to have gapless Hamiltonians, are sequences of states encoding the evolution of large, random quantum circuits.
Whereas the techniques of Crosson and Bowen~\cite{Crosson} can be used to prove spectral gap bounds for \emph{individual} history states, and with much weaker requirements on the form of the encoded circuit.
On the other hand, our probabilistic techniques are less sensitive to the form of the Hamiltonian or the structure of the history state.
So they extend to more general types of history state (\cref{def:gen_history_state}), necessary to cover some of the existing constructions in the literature, and to states with very non-uniform amplitudes.
And to larger classes of Hamiltonians, including some forms of quasi-local interactions.
It would be interesting to see if the two proof techniques can be combined in some way to get the best of both.

Our results do not rule out the possibility of a completely different way of constructing computationally hard quantum Hamiltonians with large spectral gaps.
But they do rule out any possible construction based on the only known technique in Hamiltonian complexity.
Furthermore, our proof rests on the general property that quantum circuits can quickly produce states that are locally hard to distinguish from random states.
This property would seem to represent a significant obstruction to achieving a large spectral gap with an alternative approach.
But as John Bell quipped, ``what is proved by the impossibility proofs is lack of imagination''.
We hope our results will encourage people to find alternatives to history state constructions and obsolete our results!

\clearpage

\appendix

\section{Appendix: Technical lemmas}

A key ingredient in the proof of \cref{lemproof1,lemproof2} is the following Theorem by Brandao, Harrow and Horodecki.

\begin{thm}\cite[Corollary 6]{BHH}\label{cor:main-design}
  Local random circuits of length \[425  n \lceil\log_d(4s)\rceil^2 d^2 s^5 s^{3.1/\log{(d)}}(2ns\log(d) + \log(1/\epsilon))\] form $\epsilon$-approximate $s$-designs.
\end{thm}

In the proof of \cref{lemmaconc1,lemmaconc2} we make use of the following Lemma due to Low which states that if one have a polynomial function over the unitary group with concentration with respect to Haar measure, then it will have a similar concentration when considering $\epsilon$-approximate unitary $s$-designs. As we are taking the definition of $\epsilon$-approximate $s$-designs from \cite{BHH} that differs from Low by a normalising factor, the following version of the Lemma is the one appearing in \cite{BHH}.
\begin{lem}[Low, Theorem 1.2 of \cite{Low09}]\label{Lowconcen}
  Let $f : \mathbb{U}(D) \rightarrow \mathbb{R}$ be a polynomial of degree $K$. Let $f(U) = \sum_i \alpha_i M_i(U)$ where $M_i(U)$ are monomials and let $\alpha(f) = \sum_i |\alpha_i|$. Suppose that $f$ has probability concentration
	\begin{equation*}
	\Pr_{U \sim \mu_{\text{Haar}}} \left\{ |f(U) - \mu| \geq \delta  \right\} \leq C e^{-a \delta^{2}},
	\end{equation*}
	and let $\nu_s$ be an $\epsilon$-approximate unitary $s$-design. Then for any integer $m$ with $2mK \leq s$,
	\[\Pr_{U \sim \nu_s} \left\{ |f(U) - \mu| \geq \delta  \right\} \leq v\frac{1}{\delta^{2m}} \left( C \left( \frac{m}{a}  \right)^{m} + 2\epsilon (\alpha + |\mu|)^{2m}  \right).\]
\end{lem}

We say that a set $\cN \subset \cS$ is an $\epsilon$-net of $\cS$ with respect to the distance $d$ if for every $x \in\cS$ there exists $y \in \cN$ such that $d(x,y)\leq \epsilon$. Standard arguments show \cite{MS} that
\begin{lem}\label{eps-net-hamil}
	There exists an $\epsilon$-net, $\cN_k$, of size
	$ |\cN_k|\leq \binom{m}{k} \left(\frac{3}{\epsilon}\right)^{d^{2k}}$
	of the set $\mathfrak h_k$
	with respect to the distance induced by $\|\cdot\|_\infty$.

\end{lem}

\begin{lem}\label{eps-net-circuits}
  There exists an $\epsilon$-net, $\cM_r$ of size $ |\cM_r| \leq \binom{n}{2}^r \left(\frac{6r}{\epsilon}\right)^{rd^4}$ of the set $\cC_r$ of circuits on $n$ qubits comprised of $\leq r$ two-qubit gates with respect to the  diamond norm.
\end{lem}

\begin{proof}[Proof of \cref{lemmaconc1}]
Suppose, without loss of generality, that $t_p- {t_{p'}} \geq r_1$. We start by considering the state $\ket{\psi_{t'}}$ fixed, formally, we will be just considering the probability of the function $f_h$ being bigger than a quantity conditioning on the event that whenever the time register is in the state $\ket {p'}$ the state of the computational register is $\ket{\psi_{p'}}$. Define $h':=\bra p h \ket {p'}$, then $h'$ is an operator acting only on the state register with $\|h'\|_\infty\leq \|h\|_\infty\leq 1$. Moreover, $h'=\id_{all\backslash A} \otimes h'_{A}$, where $A$ is the set of qudits from the state register where $h$ acts non-trivially.

We first assume that we are in the case where $V_p$ and $V_{p'}$ are unitaries independent of the quantum computation. For a fixed $h$ and a fixed quantum state $\ket{\psi_{p'}}$, define  $g(U)=|f_{h}(U)|^2$ with $f_{h}(U) := \bra {\psi_{p'}}V_{p'} U V_p^\dagger  h'  \ket{\psi_{p'}}$. It is easy to see that $g$ has Lipschitz constant upper bounded by $2$ and its average $\mu$ with respect to the Haar measure in $\mathbb U(d^n)$ fulfils $\mu\leq \frac {\|h'\|^2_\infty}{ d^n}\leq \frac {1}{d^n}$ (note that this function is the square of the scalar product between a Haar distributed vector of norm one and a fixed vector $V_p^\dagger h' \ket {\psi_{p'}} $ of norm $\leq \|h'\|_\infty \leq 1$). Then, by concentration over the unitary group \cite{Ledoux}, we get
\[\Pr_{U \sim \mu_{\text{Haar}}} \left\{\left||f_{h}(U)|^2 -\mu
\right| \geq  \delta  \right\} \leq 2 e^{-\frac{d^{n} \delta^{2}}{2^2 12} }. \]
From \cref{cor:main-design}, choosing
\[s_1 = \left\lfloor \left(\frac{r_1}{11050 n^2\log(d)}\right)^{1/11}\right\rfloor
\qquad\text{and}\qquad
\epsilon=\left(\frac{24s_1}{(d^{4n}+d^{-n})^2d^n}\right)^{s_1/2},\]
we get that local random quantum circuits of size $r_1$ form an $\epsilon$-approximate unitary $s_1$-design.

Now, applying \cref{Lowconcen} to function $g$. We have $D = d^{n}$, $K = 1$ and average $\mu>0$. We can upper bound $\alpha(g)\leq d^{4n}$, $\mu\leq \frac {1}{d^n}$ and Lipschitz constant upper bounded by $2$ and taking $m = s_1/2$
\begin{align*}
	\Pr_{U \sim \nu_{s_1} } \left\{\left| |f_{h}(U) |^2-\mu\right| \geq \delta  \right\} &\leq \frac{1}{\delta^{s_1}} \left( 2 \left( \frac{24 s_1}{d^{n}}  \right)^{s_1/2} +2\epsilon (d^{4n}+d^{-n})^{s_1}  \right) \nonumber \\
	&\leq 4 \left( \frac{24s_1}{d^{n}\delta^2}  \right)^{s_1/2}.
\end{align*}
For $\delta\leq 1/2, n> 3$ and using that $\mu\leq \frac {1}{d^n}$ we have
\begin{align*}\Pr_{U \sim \nu_{s_1} } \left\{|f_{h}(U) | \geq  \frac {1}{d^{n/2}}+\delta \right\} &\leq 	\Pr_{U \sim \nu_{s_1} } \left( |f_{h}(U) |^2\leq  \frac {1}{d^n}+\delta \right)\\
&\leq 4 \left( \frac{24s_1}{d^{n}\delta^2}\right)^{s_1/2}.\end{align*}

Let $\cN_k$ be a $\delta/4$-net of $\mathfrak h_k$, then we can assume (\cref{eps-net-hamil}) that
\[|\cN_k| \leq \binom{m}{k} \left(\frac{12}{\delta}\right)^{d^{2k}}.\]
Then, we have
\begin{align*}
\Pr_{U \sim \nu_{s_1} } \left\{  \max_{h\in\mathfrak h_k} |f_{h}(U) | \geq \frac {1}{d^{n/2}}+\delta \right\}&\leq\Pr_{U \sim \nu_{s_1} } \left\{  \max_{h\in\mathfrak h_k} |f_{h}(U) | \geq \frac {1}{d^{n/2}}+3\delta/4 \right\}\\
&\leq \Pr_{U \sim \nu_{s_1}} \left\{ \max_{h \in \cN_k} |f_{h}(U) | \geq \frac {1}{d^n}+\delta/2 \right\}\\
&\leq 4 \binom m k\left(\frac{12}{\delta}\right)^{d^{2k}} \left( \frac{96s_1}{d^{n}\delta^2}  \right)^{s_1/2},
\end{align*}
where we are using the fact that if $\|h-h'\|\leq \delta/4$ then $|f_h(U)-f_{h'}(U)\leq \delta/4$ in the third inequality and a union bound in the last one.
As we mention at the beginning of the proof the former probability is conditioned in the fact that whenever the time register is in the state $\ket {p'}$ the computational register is in the state $\ket {\psi_{p'}}$. But the bound on this probability is independent of the state $\ket {\psi_{p'}}$ and of the first part of the local random circuit, then we have
\begin{equation*}
\Pr_{l.r.c} \left\{  \max_{h\in \mathfrak h_k} |f_{h} | \geq \frac {1}{d^{n/2}}+\delta \right\}	\leq 4 \binom m k \left(\frac{12}{\delta}\right)^{d^{2k}} \left( \frac{96s_1}{d^{n}\delta^2}  \right)^{s_1/2}.
\end{equation*}
This finishes the proof in the case where $V_p$ and $V_{p'}$ are independent of the circuit.

If either or both of $V_p$ and $V_{p'}$ are circuits of size at most $\Vpoly(n)$, we apply a $\epsilon$-net argument on the set of circuits. We will show the case where both $V_p$ and $V_{p'}$ are circuits of size at most $\Vpoly(n)$ being the other totally analogous.

Let $\cM_{\Vpoly(n)}$ be a $\delta/8$-net of $\cC_{\Vpoly(n)}$, then we can assume (\cref{eps-net-circuits}) that
\[ |\cM_{\Vpoly(n)}| \leq \binom{n}{2}^{\Vpoly(n)} \left(\frac{48\Vpoly(n)}{\delta}\right)^{{\Vpoly(n)}d^4}.\]
Hence, for any $V,W \in \cC_{\Vpoly(n)}$, there exists a $C_V,C_W\in \cC_{\Vpoly(n)}$ such that $|\bra {\psi_{p'}}V U W^\dagger  h'  \ket{\psi_{p'}}-\bra {\psi_{p'}}C_V U C_W^\dagger  h'  \ket{\psi_{p'}}|\leq \delta/4$.

\begin{align*}
\Pr_{l.r.c.} &\left\{ \max_{h\in\mathfrak h_k} |f_{h}(U) | \geq \frac {1}{d^{n/2}}+\delta \right\}\\
&\leq\Pr_{l.r.c.} \left\{  \max_{h\in\mathfrak h_k} \max_{V_p\in \cC_r} \max_{V_{p'}\in \cC_r }|f_{h}(U) | \geq \frac {1}{d^{n/2}}+\delta \right\}\\
&\leq \Pr_{V \sim \nu_{s_1}} \left\{ \max_{h \in \cN_k}\max_{V_p\in \cM_r} \max_{V_{p'}\in \cM_r } |f_{h}(U) | \geq \frac {1}{d^n}+\delta/2 \right\}\\
&\leq 4 \binom{n}{2}^{2\Vpoly(n)} \left(\frac{48\Vpoly(n)}{\delta}\right)^{2\Vpoly(n)d^4} \binom m k\left(\frac{12}{\delta}\right)^{d^{2k}} \left( \frac{96s_1}{d^{n}\delta^2}  \right)^{s_1/2},
\end{align*}
\end{proof}

\begin{proof}[of \cref{lemmaconc2}]
Suppose, without loss of generality, that $t_p- {t_{p'}}\geq 0$. Define $h':=\bra p h \ket {p'}$, then $h'$ is an operator acting only on the state register with $\|h'\|_\infty\leq \|h\|_\infty\leq 1$. Moreover, $h'=\id_{all\backslash A} \otimes h'_{A}$, where $A$ is the set of qudits from the state register where $h$ acts non-trivially.

We first assume that we are in the case where $V_p$ and $V_{p'}$ are unitaries independent of the quantum computation. For fixed $h\in \mathfrak h_k$ and $W\in \mathbb U(n)$, define  $f=\mathrm{Re}(f_{h,W}(U))$ with $f_{h,W}(U) := \bra{0^{n}}U^{\cal y} V^\dagger_p h' V_p' W U
\ket{0^{n}}-\bra{0^{n-1}1}U^{\cal y} V^\dagger_p h' V_p' W U
\ket{0^{n-1}1}$. It is easy to see that $f$ has Lipschitz constant upper bounded by $4$ and its average $\mu$ with respect to the Haar measure in $\mathbb U(d^n)$ equals $0$, as both terms have the same average by the unitarily invariance of the Haar measure. Then, by concentration over the unitary group \cite{Ledoux},
\[\Pr_{U \sim \mu_{\text{Haar}}} \left\{|\mathrm{Re}( f_{h,W}(U) )| \geq \frac \delta {2\sqrt 2} \right\} \leq 2 e^{-\frac{d^{n} (\delta/2\sqrt 2)^{2}}{4^2 12} }.\]

From \cref{cor:main-design}, choosing
\[s = \left\lfloor \left(\frac{r}{1900 n^2\log(d)}\right)^{1/11}\right\rfloor
\qquad\text{and}\qquad
\epsilon= \left(\frac{768s}{(2d^{2n})^2d^n}\right)^{s/2},\]
we have that local random quantum circuits of size $r$ form an $\epsilon$-approximate unitary $s$-design.
Now, applying \cref{Lowconcen}, we want to turn the concentration of $f$ over the unitary group to concentration over the set of random quantum circuits of size $r$.	For the function $f$, we have that  $D = d^{n}$, $\mu = 0$ and the degree is $K=1$. We can upper bound $\alpha(f)\leq 2\sum_{ab=1}^{d^n} |(h'V)_{ab}|\leq 2d^{2n}$. Then using \cref{Lowconcen} with $m = s/2$ we get
\begin{align*}
	\Pr_{U \sim \nu_s } \left\{|\mathrm{Re}( f_{h,W}(U) )| \geq \frac \delta {2\sqrt 2} \right\}
	&\leq
	\frac{1}{(\delta/2\sqrt 2)^{s}}
	\left( 2 \left( \frac{768s}{d^{n}}  \right)^{s/2} +
	2\epsilon (2d^{2n})^{s}  \right) \nonumber \\
	&= 4 \left( \frac{6144 s}{d^{n}\delta^2}  \right)^{s/2}.
\end{align*}

Repeating the same argument for the function $\mathrm{Im}(f_{h,V}(U))$ and applying a union bound argument we get
\[\Pr_{U \sim \nu_s } \left\{| f_{h,W}(U) | \geq \frac \delta {2} \right\}\leq 8 \left( \frac{6144 s}{d^{n}\delta^2}  \right)^{s/2}.\]

Let $\cN_k$ be a $\delta/4$-net of $\mathfrak h_k$, then we can assume (\cref{eps-net-hamil}) that
\[|\cN_k| \leq \binom{m}{k} \left(\frac{12}{\delta}\right)^{d^{2k}}\]
and applying a union bound argument, we have that
\begin{align*}
\Pr_{U \sim \nu_s } \left\{  \max_{h\in\mathfrak h_k}  |f_{h,V}(U) | \geq \delta \right\} &\leq \Pr_{U \sim \nu_s } \left\{  \max_{h\in\mathfrak h_k}  |f_{h,V}(U) | \geq 3\delta/4 \right\}\\
&\leq \Pr_{U \sim \nu_s} \left\{ \max_{h \in \cN_k} |f_{h,V}(U)| \geq \delta/2 \right\}\\
&\leq 8\binom m k \left(\frac{12}{\delta}\right)^{d^{2k}}  \left( \frac{6144 s}{d^{n}\delta^2}  \right)^{s/2}.
\end{align*}

Finally, the bound on the probability is independent of $V$. Hence, we have that \[\Pr_{l.r.c.} \left\{ \max_{h\in\mathfrak h_k}  |f_{h}| \geq \delta \right\} \leq 8 \binom m k \left(\frac{12}{\delta}\right)^{d^{2k}}  \left( \frac{6144 s}{d^{n}\delta^2}  \right)^{s/2}.\]
This finishes the proof in the case where $V_p$ and $V_{p'}$ are independent of the circuit.

If either or both of $V_p$ and $V_{p'}$ are circuits of size at most $\Vpoly(n)$, we can take $\cM_{\Vpoly(n)}$ a $\delta/8$-net of the set $\cC_{\Vpoly(n)}$ as in the proof of \cref{lemmaconc1}, and repeating the reasoning there we get the result.
\end{proof}

\section*{Acknowledgements}

The authors thank Elizabeth Crosson for illuminating discussions about these results.
CEGG thanks UCL Department of Computer Science and the CS Quantum group for their hospitality during the first semester of 2017.
TSC is supported by the Royal Society.
CEGG is supported by Spanish MINECO (project MTM2017-88385-P) and MECD ``Jos\'e Castillejo'' program (CAS16/00339).

\printbibliography

\end{document}